\documentclass[onecolumn,draft, 11pt]{IEEEtran}

\usepackage{color}

\usepackage{graphicx,tabularx,array,geometry,amsmath,amsthm,thmtools}

\usepackage{mathtools}
\DeclarePairedDelimiter{\ceil}{\lceil}{\rceil}

\usepackage{caption}
\usepackage{amsfonts}
\usepackage{bbm}
\usepackage{makecell}
\usepackage{multirow}
 \usepackage{amssymb}
\usepackage{amsthm}
\usepackage{txfonts}
\usepackage[T1]{fontenc}
\usepackage{tikz}
\usepackage[scr=dutchcal]{mathalfa}
\let\mathscr\mathbscr

\newcolumntype{x}[1]{>{\centering\arraybackslash}p{#1}}

\declaretheoremstyle[headfont=\bfseries, 
    bodyfont=\normalfont]{normalhead}
\declaretheorem[style=normalhead]{Example}

\newtheorem{Theorem}{Theorem}
\newtheorem{Proposition}{Proposition}
\newtheorem{Lemma}{Lemma}
\newtheorem{Corollary}{Corollary}
\newtheorem{Remark}{Remark}
\newtheorem{Definition}{Definition}

\newtheorem{Claim}{Claim}

\allowdisplaybreaks 

\begin{document}

\title{On the Sub-optimality of Single-Letter Coding over Networks}
\author{Farhad Shirani  and S. Sandeep Pradhan\footnote{This work was supported by NSF grant CCF-1422284. This work will be presented in part at IEEE International Symposium on Information Theory (ISIT), July 2017.} \\
Dept. of Electrical Engineering and Computer Science \\
 Univ. of  Michigan, Ann Arbor, MI. \\\date{} }

\maketitle \thispagestyle{empty} \pagestyle{plain}

\begin{abstract}

In this paper, we establish a new bound tying together the effective length and the maximum correlation between the outputs of an arbitrary pair of Boolean functions which operate on two sequences of correlated random variables. We derive a new upper bound on the correlation between the outputs of these functions. The upper bound may find applications in problems in many areas which deal with common information. We build upon Witsenhausen's result \cite{ComInf2} on maximum correlation.  The present upper bound takes into account the effective length of the Boolean functions in characterizing the correlation. 

We use the new bound to characterize the communication-cooperation tradeoff in multi-terminal communications. We investigate binary block-codes (BBC). A BBC is defined as a vector of Boolean functions. We consider an ensemble of BBCs which is randomly generated using single-letter distributions. We characterize the vector of dependency spectrums of these BBCs. We use this vector to bound the correlation between the outputs of two distributed BBCs. Finally, the upper bound is used to show that the large blocklength single-letter coding schemes studied in the literature are sub-optimal in various multi-terminal communication settings.
\end{abstract}


\section{Introduction}

In his paper, "A Mathematical Theory of Communications" \cite{Shannon} - often regarded as the Magna Carta of digital communications - Shannon pointed out that in order to exploit the redundancy of the source in data compression, it is necessary to compress large blocks of the source simultaneously. More precisely, optimality is only approached as the  blocklength of the encoding functions approaches infinity. The same observation was made in the case of point-to-point (PtP) channel coding. As a result, a common feature of the coding schemes used in PtP communication is that they have large blocklengths.  In the source coding problem, by compressing large blocks at the same time, one can exploit the redundancy in the source. In the channel coding problem, transmitting the input message over large blocks allows the decoder to exploit the typicality of the noise vector.

The blocklength of an encoding function is defined as the length of its input sequence. An additional parameter of interest is the \textit{effective-length} of the encoding function. The effective-length can be interpreted as the average number of input elements necessary to determine an output element. It can be shown that in order to achieve optimality in PtP communications, the effective-length of the encoding function must be asymptotically large. Loosely speaking, this means that each output element of the encoding function must be a function of the entire input sequence, and the length of the input sequence must increase asymptotically. As a result, in addition to having large blocklengths, the encoding functions generated by Shannon's coding scheme have large {effective-lengths}.

There is a subtle difference between the blocklength of an encoding function and its effective-length. As an example, assume that $\underline{e}(X^n)$ is an encoding function with blocklength and effective-length both equal to $n$. This means that the encoding function takes input sequences of length $n$ and its output elements are functions of the entire input sequence. Define a new encoding function $\underline{f}(X^{2n})$ as the concatenation of $\underline{e}$ with itself (i.e. $\underline{f}(X^{2n})=(\underline{e}(X^n),\underline{e}(X_{n+1}^{2n}))$. Then, $\underline{f}(X^{2n})$ has blocklength equal to $2n$ and effective-length equal to $n$. Generally, concatenation of encoding functions increases blocklength, but does not affect effective length. 

 Similar to PtP communications, in multi-terminal communications, it is well-known that the performance of block-encoders is super-additive with respect to blocklength, i.e., the best performance of block-encoders of a certain length is an increasing function of the blocklength. However, in this work we show that the performance as a function of the effective-length of the encoder is not super-additive. This phenomenon can be explained as follows. In multiterminal communication it is often desirable to maintain correlation among the output sequences at different nodes. This requirement can be due to explicit constraints in the problem statement such as joint distortion measures in multi-terminal source coding, or it can be due to implicit factors such as the need for interference reduction, or the nature of the shared communication medium in multi-terminal channel coding. In the latter case, the correlation between the outputs is necessary as a means to further the cooperation among the transmitters. In this paper, we show that pairs of encoding functions with large effective-lengths are inefficient in coordinating their outputs. This is due to the fact that such encoding functions are unable to produce highly correlated outputs from highly correlated inputs. The loss of correlation undermines the encoders' ability to cooperate and take advantage of the multi-terminal nature of the problem. In PtP communication problems, where there is only one transmitter, the necessity for cooperation does not manifest itself. For this reason, although encoders with asymptotically large effective-lengths are optimal in PtP communications, they are sub-optimal in various network communication settings.

 In this work we make the following contributions. First, we consider two Boolean functions of two correlated discrete, memoryless sources  (DMS). We {quantify} the correlation between these functions {using} the probability that their outputs are equal. We derive an upper-bound on this correlation. The bound is presented as a function of the effective-length of the two Boolean functions. Next, we consider two arbitrary binary block codes (BBC) as defined in \cite{ComInf1}. A BBC is a vector of Boolean functions.  The two encoding functions are applied to two correlated DMSs. We quantify the correlation between these encoding functions using the average probability that any two output-bits are equal, where the average is over the elements of the output vector.  We use the bound from the previous step, to show that  the code ensembles generated by the standard large blocklength single-letter coding schemes used in multi-terminal communication problems do not produce highly correlated outputs from highly correlated inputs. We prove that there is a discontinuity in the correlation preserving abilities of encoding functions produced using single-letter coding schemes. The single-letter coding schemes considered in this work are general and include Shannon's point-to-point source coding scheme \cite{Shannon}, the Berger-Tung scheme \cite{Markov}, the Zhang-Berger scheme \cite{ZB}, the Cover-El Gamal- Salehi scheme \cite{CES}, and the Salehi-Kurtas scheme \cite{ICcorr}. Lastly, we show that this discontinuity leads to the sub-optimality of single-letter coding schemes in various multi-terminal communication problems.

%
%
%

In the first step, we characterize the maximum correlation between the outputs of a pair of Boolean functions of random sequences. This is a fundamental problem of broad theoretical and practical interest. Consider the two distributed agents shown in Figure \ref{fig:agents}.
A pair of correlated discrete memoryless sources (DMS) are fed to the two agents. These agents are to each make a binary decision. The goal of the problem is to maximize the correlation between the outputs of these agents subject to specific constraints on the decision functions. The study of this setup has had impact on a variety of disciplines, for instance, by taking the agents to be two encoders in the distributed source coding problem \cite{FinLen}, or two transmitters in the interference channel problem, or Alice and Bob in a secret key-generation problem \cite{security2, security3}, or two agents in a distributed control problem \cite{control}. 

A special case of the problem is the study of common-information (CI) generated by the two agents. As an example, consider two encoders in a Slepian-Wolf (SW) distributed compression setup \cite{SW}. Let $U_1,U_2$, and $V$ be independent, binary random variables.
\begin{figure}[!t]
\centering
\includegraphics[height=1.2in, draft=false]{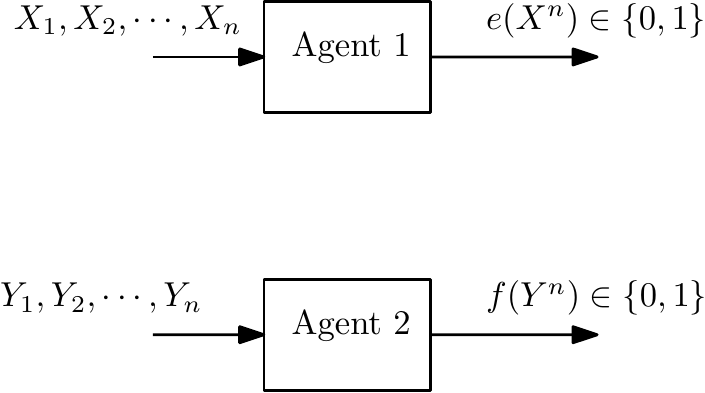}
 \caption{Correlated Boolean decision functions.}
\label{fig:agents}
\end{figure}
Then, an encoder observing the DMS $X=(V,U_1)$, and an encoder observing $Y=(V,U_2)$ agree on the value of $V$ with probability one. The random variable $V$ is called the CI observed by the two encoders. These encoders can transmit the sources to the joint decoder with a sum-rate equal to $H(V)+H(U_1)+H(U_2)$. This gives a reduction in rate equal to the entropy of $V$, as compared to the independent compression of the sources. The gain in performance is directly related to the entropy of the CI. So, it is desirable to maximize the entropy of the CI between the encoders.
 
  In \cite{ComInf1}, the authors investigated multi-letterization as a method for increasing the CI. They showed that multi-letterization does not lead to an increase in the CI. More precisely, they prove the following statement: 
  
{ \textit{Let $X$ and $Y$ be two DMSs. Let $(X^n,Y^n)$ denote $n$ independent copies of the sources. Let $f_{n}(X^n)$ and $g_{n}(Y^n)$ be two sequences of functions which converge to one another in probability. Then, the normalized entropies $\frac{1}{n}H(f_{n}(X^n))$, and $\frac{1}{n}H(g_{n}(Y^n))$ are less than or equal to the entropy of the CI between $X$ and $Y$ for large $n$. }}
 
 A stronger version of the result was proved by Witsenhausen \cite{ComInf2}, where the maximum correlation between the outputs was bounded from above subject to the restriction that the entropy of the binary output is fixed. 
 It was shown that the maximum correlation is achieved if both users output a single element of the sequence without further processing (e.g. each user outputs the first element of its corresponding sequence). This was used to conclude that common information can not be increased by multi-letterization. While, the result has been used extensively in a variety of areas such as information theory, security, and control \cite{security2,security3, control, dictator_fun}, in many problems, there are additional constraints on the set of admissible decision functions. For example, one can consider constraints on the effective-length of the decision functions. These are natural constraints, for instance, in the case of communication systems, the users have lower bounds on their effective-lengths due to the rate-distortion requirements in the problem. 
 
 In this paper, the problem under these additional constraints is considered. 
Initially, we define the effective-length for additive Boolean functions. For non-additive functions, we use a method similar to \cite{ComInf2}, and map the Boolean functions to the set of real-valued functions. Using tools in functional analysis, we find an additive decomposition of these functions. The decomposition components have well-defined effective-lengths. Using the decomposition we find the dependency spectrum of the Boolean function. The dependency spectrum is a generalization of the effective-length and is defined for non-additive Boolean functions. Next, a new upper bound on the correlation between the outputs of arbitrary pairs of Boolean functions is derived. The bound is presented as a function of the dependency spectrum of the Boolean functions.  

In the next step, we show that the loss in correlation caused by the application of large effective-length codes causes a discontinuity in the performance of schemes using such codes in some multi-terminal problems. This was first observed in the Berger-Tung achievable rate-distortion region for the problem of distributed source coding \cite{Jchendisc} \cite{wagner}. It was noted that when the common information is available to the two encoders in the distributed source coding problem, the performance is discontinuously better than when the common information is replaced with highly correlated components. In \cite{FinLen}, we argued that the discontinuity in performance is due to the fact that the encoding functions in the Berger-Tung scheme preserve common information, but are unable to preserve correlation between highly correlated components. We proposed a new coding scheme, and derived an improved achievable rate-distortion region for the two user distributed source coding problem \cite{FinRD}. The new strategy uses a concatenated coding scheme which consists of one layer of codes with finite effective-length, and one layer of codes with asymptotically large effective-lengths.

{In this paper, we identify the underlying phenomenon which leads to the sub-optimality of the Berger-Tung coding scheme. We show that the phenomenon is not restricted to the distributed source coding problem, rather it is observed in a variety of coding schemes which are used for the distributed transmission of sources over channels, and the multiple descriptions source coding problem. The communication setups include the transmission of correlated sources over the multiple-access channel, as well as the interference channel. Furthermore, we derive a new necessary condition for optimality of coding schemes used in these problems. The necessary condition is that the encoding functions used for the distributed transmission of correlated sources must have effective-lengths which are bounded from above (i.e. their lengths are not asymptotically large.). It is shown that the standard coding strategies used in the literature have sub-optimal performance since they do not satisfy this necessary condition.  }

%
%
%
%

%
%
%
The rest of the paper is organized as follows: In Section \ref{sec:not}, we explain the notation used in this paper. Section \ref{sec:eff} develops the tools used for analyzing Boolean functions. In Section \ref{sec:corr}, we present the bound which relates the effective length and the maximum correlation between two Boolean functions of sequences of random variables. In Section \ref{sec:SLCS}, it is shown that encoding functions generated using single-letter coding ensembles have large effective lengths. Section \ref{sec:ex} contains three examples where finite effective length is necessary to achieve optimal performance in multi-terminal communications. Finally, Section \ref{sec:conc} concludes the paper. 
\section{Notation}\label{sec:not}
In this section, we introduce the notation used in this paper. We represent random variables by capital letters such as $X, U$. Sets are denoted by calligraphic letters such as $\mathcal{X}, \mathcal{U}$.  Particularly, the set of natural numbers and real numbers are shown by $\mathbb{N}$, and $\mathbb{R}$, respectively. For a random variable $X$, the corresponding probability space is $(\mathcal{X}, \mathbf{F}_{X}, P_X)$, where $\mathbf{F}$ is the underlying $\sigma$-field. The set of all subsets of $\mathcal{X}$ is written as $2^{\mathcal{X}}$. There are three different notations used for different classes of vectors. For random variables, the $n$-length vector $(X_1,X_2,\cdots,X_n), X_i\in \mathcal{X}$ is denoted by $X^n\in \mathcal{X}^n$. For the vector of functions $(e_1(X),e_2(X),\cdots, e_n(X))$ we use the notation $\underline{e}(X)$. The binary string $(i_1,i_2,\cdots,i_n), i_j\in \{0,1\}$ is written as $\mathbf{i}$.  As an example, the set of functions $\{\underline{e}_{\mathbf{i}}(X^n)| \mathbf{i}\in \{0,1\}^n\}$ is the set of $n$-length vectors of functions $(e_{1,\mathbf{i}},e_{2,\mathbf{i}},\cdots,e_{n,\mathbf{i}})$ operating on the vector $(X_1,X_2,\cdots,X_n)$ each indexed by an $n$-length binary string $(i_1,i_2,\cdots,i_n)$. The vector of binary strings $(\mathbf{i}_1,\mathbf{i}_2,\cdots, \mathbf{i}_n)$ denotes the standard basis for the $n$-dimensional space (e.g. $\mathbf{i}_1=(0,0,\cdots,0,1)$). The vector of random variables $(X_{j_1},X_{j_2},\cdots, X_{j_k}), j_i\in [1,n], j_i\neq j_k$, is denoted by $X_{\mathbf{i}}$, where $i_{j_l}=1, \forall l\in [1,k]$. For example, take $n=3$, the vector $(X_1,X_3)$ is denoted by $X_{101}$, and the vector $(X_1,X_2)$ by $X_{110}$.  Particularly, $X_{\mathbf{i}_j}=X_j , j\in [1,n]$. Also, for $\mathbf{t}=\underline{1}$, the all-ones vector, $X_{\mathbf{t}}=X^n$. For two binary strings $\mathbf{i},\mathbf{j}$, we write $\mathbf{i}\leq\mathbf{j}$ if and only if $i_k\leq j_k, \forall k\in[1,n]$. Also, we write $\mathbf{i}<\mathbf{j}$ if $\mathbf{i}\leq\mathbf{j}$ and $\mathbf{i}\neq\mathbf{j}$.
 For a binary string $\mathbf{i}$ we define $N_{\mathbf{i}}\triangleq w_H(\mathbf{i})$, where $w_H$ denotes the Hamming weight. Lastly, the vector $\sim \mathbf{i}$ is the element-wise complement of $\mathbf{i}$.


\section{The \textit{effective-length} of an Encoder}\label{sec:eff}
In this section, we define a set of parameters which measure the {effective-length} of an encoding function. This is done in several steps. First, we characterize the effective-length of the additive Boolean functions. Then, for general Boolean functions, we find a decomposition of these functions into additive functions, and define a generalization of the effective-length called the `\textit{dependency spectrum}' of a Boolean function. 
\subsection{Additive Binary Block Encoders}
We only consider encoders with binary outputs\footnote{The analysis provided in this paper can be generalized to arbitrary finite output alphabets. The interested reader can refer to Section 7 in \cite{ComInf2}. }. The encoder can be viewed as a vector of Boolean functions. The following provides the definition of the effective-length of an additive Boolean function:
\begin{Definition}
 Let $\mathcal{X}$ be a finite set, and $n$ an integer. For a Boolean function $e:\mathcal{X}^n\to \{0,1\}$ defined by $e(X^n)=\sum_{i\in \mathsf{J}}X_i, \mathsf{J}\subset [1,n] $, where the addition operator is the binary addition, the effective-length is defined as the cardinality of the set $\mathsf{J}$.
\end{Definition}
The effective-length of an encoding function - which is a vector of additive Boolean functions - is defined as the average of the effective-lengths of each of its elements. In order to clarify this definition, let us look at the following example.
\begin{Example}
 Let $\mathcal{X}=\{0,1\}$. Consider an encoding function $\underline{e}$ which takes $n$-length blocks of the input $X$ and produces output elements $e_{i}(X^n), i\in[1,n]$ using the following mapping:
\begin{align}
{e}_i(X^n)=    \begin{cases}
      X_i\oplus_2 X_{i+1}, & i\neq n, \\
      X_n\oplus_2X_1. & i=n.
    \end{cases}
\end{align}
Clearly, each output element is completely determined by the value of $X_i$ and $X_{i\oplus_n 1}$. Based on our interpretation, each component of $\underline{e}$ has  {effective-length} equal to two. 
\end{Example}
In the above example, the effective-length has a scalar value, and the value is equal to the minimum number of input elements required to compute the output element. It is tempting to extend this definition to arbitrary functions. 
For a generic encoding function, every output $e_i(X^n), i\in [1,n]$ depends on most of the elements in $X^n$. Hence, 
the value of the effective length as defined above would be essentially trivial and of little use. For an arbitrary encoding function, it would be more meaningful to ask questions such as how strongly does the first element $X_1$ affect the output of $e_i(X^n)$? Is this effect amplified when we take $X_2$ into account as well? 
Is there a subset of random variables that (almost) determines the value of the output? In this section, we formulate these questions in mathematical terms. We define the {dependency spectrum} as a vector which captures the correlation between different subsets of the input elements with each element of the output. The dependency spectrum can be viewed as a generalization of the effective-length.

   
   We proceed by formally defining the problem. We assume that two correlated DMS's are being fed to two arbitrary encoders, and analyze the correlation between the outputs of these encoders. The following gives the formal definition for DMS's.
 \begin{Definition}
  $(X,Y)$ is called a pair of DMS's if we have $P_{X^n,Y^n}(x^n,y^n)=\prod_{i\in[1,n]}P_{X_i,Y_i}(x_i,y_i), \forall n\in \mathbb{N}, x^n\in \mathcal{X}^n,y^n\in \mathcal{Y}^n$, where $P_{X_i,Y_i}=P_{X,Y}, \forall i\in [1,n]$, for some joint distribution $P_{X,Y}$.
\end{Definition}
 Akin to the results presented in \cite{ComInf2} and \cite{ComInf1}, we restrict our attention to the binary block encoders (BBE), which are defined below. 
\begin{Definition}
 A Binary-Block-Encoder is characterized by the triple $(\underline{e},\mathcal{X},n)$, where $\underline{e}$ is a mapping $\underline{e}:\mathcal{X}^n\to \{0,1\}^n$, $\mathcal{X}$ is a finite set, and $n$ is an integer.  
\end{Definition}
We refer to a BBE by its corresponding mapping $\underline{e}$. The mapping $\underline{e}$ can be viewed as a vector of functions $(e_i)_{i\in [1,n]}$, where $e_i:\mathcal{X}^n\to\{0,1\}$. 
We convert the problem of analyzing a BBE into one where the encoder is a binary real-valued function. Converting the discrete-valued encoding function into a real-valued one is crucial since it allows us to use the rich set of tools available in functional analysis. We present a summary of the functional analysis apparatus used in this work. 

\subsection{Real Transformations of Boolean Functions}
Fix a discrete memoryless source $X$, and a BBE defined by $\underline{e}:\mathcal{X}^n\to \{0,1\}^n$.  Let $P\left(e_i(X^n)=1\right)=q_i$. The real-valued function corresponding to $e_i$ is represented by $\tilde{e}_i$, and is defined as follows:
\begin{align}
\tilde{e}_i(X^n)=    \begin{cases}
      1-q_i, & \qquad  \text{if } e_i(X^n)=1, \\
      -q_i. & \qquad\text{otherwise}.
    \end{cases}
\end{align}
\begin{Remark}
 Note that $\tilde{e}_i, i\in [1,n]$ has zero mean and variance $q_i(1-q_i)$.
\label{Rem:exp_0}
\end{Remark}
The random variable $\tilde{e}_i(X^n)$ has finite variance on the probability space $(\mathcal{X}^n, 2^{\mathcal{X}^n}, P_{X^n})$. The set of all such functions is denoted by $\mathcal{H}_{X,n}$. More precisely,  we define $\mathcal{H}_{X,n}\triangleq L_2(\mathcal{X}^n, 2^{\mathcal{X}^n}, P_{X^n})$ as the separable Hilbert space of all measurable functions $\tilde{h}:\mathcal{X}^n\to \mathbb{R}$ with inner product given by $\tilde{h}\cdot\tilde{g}=\sum_{x^n}\tilde{h}(x^n)\tilde{g}(x^n)P_{X^n}(x^n)$.  Since X is a DMS, the isomorphism relation 
 \begin{equation}
 \mathcal{H}_{X,n}= \mathcal{H}_{X,1}\otimes \mathcal{H}_{X,1}\cdots \otimes \mathcal{H}_{X,1}
\label{eq:Hil_Dec1}
 \end{equation}
 holds \cite{Reed_and_Simon}, where $\otimes$ indicates the tensor product.

\begin{Example}
Let $n=1$. Let $\mathcal{X}=\{0,1\}$. The Hilbert space $\mathcal{H}_{X,1}$ is the space of all measurable functions $\tilde{h}:\mathcal{X}\to \mathbb{R}$. The space is spanned by the two linearly independent functions  $\tilde{h}_1(X)=\mathbbm{1}_{\{X=1\}}$ and $\tilde{h}_2(X)=\mathbbm{1}_{\{X=0\}}$. We conclude that the space is two-dimensional.
\label{Ex:ex1}  
\end{Example}

As a reminder, the following defines the tensor product of vector spaces.

\begin{Definition}[\cite{Reed_and_Simon}]
 Let $\mathcal{H}_{i},i \in[1,n]$ be vector spaces over a field $F$. Also, let $\mathcal{B}_{i}=\{v_{i,j}|j\in [1,d_i]\}$ be the basis for $\mathcal{H}_i$ where $d_i$ is the dimension of $\mathcal{H}_i$. Then, the tensor product space $ \otimes_{i\in [1,n]}\mathcal{H}_i$ is defined as the set of elements $v=\sum_{j_1\in [1,d_1]}\sum_{j_2\in [1,d_2]}\cdots \sum_{j_n\in [1,d_n]} c_{j_{1}, j_2,\cdots,j_n} v_{j_1}\otimes v_{j_2}\cdots \otimes v_{j_n}$.
 \label{Lem:tensor_dec}
\end{Definition}

 \begin{Remark}
  The tensor product operation in $\mathcal{H}_{X,n}$ is real multiplication (i.e. $f_1, f_2\in \mathcal{H}_{X,1}: f_1(X_1)\otimes f_2(X_2)\triangleq f_1(X_1)f_2(X_2)$). So, if $\{f_i(X)|i\in [1,d]\}$ is a basis for  $\mathcal{H}_{X,1}$ when $|\mathcal{X}|=d$, a basis for  $\mathcal{H}_{X,n}$ would be the set of all the real multiplications of these basis elements: $\{\Pi_{j\in [1,n]}f_{i_j}(X_j), i_j\in [1,d]\}$. 
\end{Remark}

 Example \ref{Ex:ex1} gives a decomposition of the space $\mathcal{H}_{X,1}$ for binary input alphabets. Next, we introduce a decomposition of $\mathcal{H}_{X,1}$ for general alphabets which turns out to be very useful. Let $\mathcal{I}_{X,1}$ be the subset of all measurable functions of $X$ which have 0 mean, and let $\gamma_{X,1}$ be the set of constant real functions of $X$.  $\mathcal{I}_{X,1}$ and  $\gamma_{X,1}$  are linear subspaces of $\mathcal{H}_{X,1}$. $\mathcal{I}_{X,1}$ is the null space of the functional which takes an arbitrary function $\tilde{f}\in \mathcal{H}_{X,1}$ to its expected value $\mathbb{E}_{X}(\tilde{f})$. The null space of any non-zero linear functional is a hyper-space in $\mathcal{H}_{X,1}$. So, $\mathcal{I}_{X,1}$ is a $(|\mathcal{X}-1|)$-dimensional subspace of $\mathcal{H}_{X,1}$. 
 On the other hand, $\gamma_{X,1}$ is a one dimensional subspace which is not contained in $\mathcal{I}_{X,1}$. It is spanned by the function $\tilde{g}(X)\equiv 1$.
   Consider an arbitrary element $\tilde{f}\in \mathcal{H}_{X,1}$. One can write $\tilde{f}= \tilde{f}_1+\tilde{f}_2$ where $\tilde{f}_1=\tilde{f}-\mathbb{E}_{X}(\tilde{f})\in \mathcal{I}_{X,1}$, and $\tilde{f}_2= \mathbb{E}_{X}(\tilde{f})\in  \gamma_{X,1}$.  Hence, $\mathcal{H}_{X,1}=\mathcal{I}_{X,1}\oplus \gamma_{X,1}$ gives a decomposition of $\mathcal{H}_{X,1}$.
 \label{Ex:one_dim}
Replacing $\mathcal{H}_{X,1}$ with $\mathcal{I}_{X,1}\oplus \gamma_{X,1}$ in \eqref{eq:Hil_Dec1}, we have:
 \begin{align}
 \mathcal{H}_{X,n}&=\otimes_{i=1}^n \mathcal{H}_{X,1}=\otimes_{i=1}^n (\mathcal{I}_{X,1}\oplus \gamma_{X,1})\stackrel{(a)}{=} \oplus_{\mathbf{i}\in \{0,1\}^n} (\mathcal{G}_{i_1}\otimes \mathcal{G}_{i_2}\otimes\dotsb \otimes \mathcal{G}_{i_n}),
 \label{eq:Hil_Dec2}
\end{align}
where
\begin{align*}
 \mathcal{G}_j=
\begin{cases}
 \gamma_{X,1}  \qquad \ j=0,\\
 \mathcal{I}_{X,1} \qquad \ j=1,
\end{cases}
\end{align*}
and, in (a), we have used the distributive property of tensor products over direct sums. 
Using equation \eqref{eq:Hil_Dec2} we can define the following:
\begin{Definition}
 For any $\tilde{e}\in \mathcal{H}_{X,n},  n\in \mathbb{N}$, define the decomposition \footnote{It turns out that for a symmetric source $X$, this decomposition is similar to the Fourier transform of Boolean functions \cite{Odonnel}.}
  $\tilde{e}=\sum_{\mathbf{i}}\tilde{e}_{\mathbf{i}}$, where $\tilde{e}_{\mathbf{i}}\in \mathcal{G}_{i_1}\otimes \mathcal{G}_{i_2}\otimes\dotsb \otimes \mathcal{G}_{i_n}$. $\tilde{e}_{\mathbf{i}}$ is the component of $\tilde{e}$ which is only a function of $\{X_{i_j}|i_j=1\}$. In this sense, the collection $\{\tilde{e}_{\mathbf{i}}|\sum_{j\in[1,n]}i_j=k\}$, is the set of k-letter components of $\tilde{e}$.
\label{Rem:Dec}
\end{Definition}
In order clarify the notation, we provide the following two examples.

\begin{Example}
 Let $(X_1,X_2)$ be two independent symmetric binary random variables. Assume $e(X_1,X_2)=X_1\oplus X_2$ is the binary addition function. In this example $P(e=1)=\frac{1}{2}$. The corresponding real function is given as follows:
\begin{align*}
 \tilde{e}(X_1,X_2)=
\begin{cases}
- \frac{1}{2}  \qquad \ X_1+X_2\in \{0,2\},\\
\frac{1}{2} \qquad \ X_1+X_2=1,
\end{cases}
\end{align*} 
 Using Lagrange interpolation \cite{lagrange}, we can write $\tilde{e}$ as follows:
 \begin{align*}
 \tilde{e}&=-\frac{1}{2}(X_1+X_2-2)(X_1+X_2)-\frac{1}{4}(X_1+X_2-1)(X_1+X_2-2)-\frac{1}{4}(X_1+X_2)(X_1+X_2-1)
 \\&=-X_1^2-X_2^2-2X_1X_2+2X_1+2X_2-\frac{1}{2}. 
\end{align*}
The decomposition of $\tilde{e}$ in the form given in \eqref{eq:Hil_Dec2} is
\begin{align*}
 &\tilde{e}_{1,1}= X_1+X_2-2X_1X_2-\frac{1}{2}=-\frac{1}{2}(1-2X_1)(1-2X_2),
 \\& \tilde{e}_{1,0}=-X_1^2+X_1=X_1(1-X_1)\stackrel{(a)}{=}0,  \quad \tilde{e}_{0,1}= -X_2^2+X_2=X_2(1-X_2)\stackrel{(a)}{=}0,
\\&\tilde{e}_{0,0}=0.
\end{align*} 
where (a) holds since the input is chosen from $\{0,1\}$.  
Note that $\tilde{e}$ has a single non-zero component in its decomposition. This component is the two-letter function $\tilde{e}_{1,1}\in \mathcal{I}_{X,1}\otimes \mathcal{I}_{X,1}$. This is to be expected since the binary addition of two symmetric variables is independent of each variable. So there are no single-letter components. In fact one can verify this directly as follows:
\begin{align*}
 \mathbb{E}_{X_2|X_1}(\tilde{e}|X_1)=X_1-X_1=0,\quad  \mathbb{E}_{X_1|X_2}(\tilde{e}|X_2)=X_2-X_2=0. 
\end{align*}
\begin{Remark}
 In the previous example, we found that the binary summation of two independent binary symmetric variables is a two-letter function (i.e. it only has a two-letter component). 
  However, this is not true when the source is not symmetric. When $P(X=1)\neq P(X=0)$, the output of the summation is not independent of each of the inputs. One can show that the single-letter components of the summation are non-zero in this case. 
\end{Remark}
\end{Example}
\begin{Example}
 Let $e(X_1,X_2)=X_1\wedge X_2$ be the binary logical AND function. The corresponding real function is:
 \begin{align*}
 \tilde{e}(X_1,X_2)=
\begin{cases}
 -\frac{1}{4}  \qquad \ (X_1,X_2)\neq (1,1),\\
\frac{3}{4} \qquad \ (X_1,X_2)=(1,1).
\end{cases}
\end{align*} 
Lagrange interpolation gives $\tilde{e}=X_1X_2-\frac{1}{4}$. The decomposition is given by:
\begin{align*}
 &\tilde{e}_{1,1}= (X_1-\frac{1}{2})(X_2-\frac{1}{2}),\qquad \tilde{e}_{1,0}=\frac{1}{2}(X_1-\frac{1}{2}), \qquad \tilde{e}_{0,1}= \frac{1}{2}(X_2-\frac{1}{2}),\qquad\tilde{e}_{0,0}=0.
\end{align*} 
The variances of these functions are given below:
\begin{align*}
 Var(\tilde{e})=\frac{3}{16}, \qquad Var(\tilde{e}_{0,1})=Var(\tilde{e}_{1,0})=Var(\tilde{e}_{1,1})=\frac{1}{16}. 
\end{align*}
As we shall see in the next sections, these variances play a major role in determining the correlation preserving properties of $\tilde{e}$. In the perspective of the effective-length, the function $\tilde{e}$ has $\frac{2}{3}$ of its variance distributed between $\tilde{e}_{0,1}$, and $\tilde{e}_{1,0}$ which are single-letter functions, and $\frac{1}{3}$ of the variance is on $\tilde{e}_{1,1}$ which is a two-letter function. 
 
\end{Example}
Similar to the previous examples, for arbitrary $\tilde{e}\in \mathcal{H}_{X,n},  n\in \mathbb{N}$, we can find a decomposition $\tilde{e}=\sum_{\mathbf{i}}\tilde{e}_{\mathbf{i}}$, where $\tilde{e}_{\mathbf{i}}\in \mathcal{G}_{i_1}\otimes \mathcal{G}_{i_2}\otimes\dotsb \otimes \mathcal{G}_{i_n}$. We can characterize $\tilde{e}_{\mathbf{i}}$ in terms of products of the basis elements of $\mathcal{G}_{i_1}\otimes \mathcal{G}_{i_2}\otimes\dotsb \otimes \mathcal{G}_{i_n}$ as follows.

\begin{Lemma}
\label{lem:dec_non_bin}
  For an arbitrary input alphabet $\mathcal{X}$, let $\tilde{h}_l(X), l\in \{1,2,\cdots, |\mathcal{X}|-1\}$ be an orthogonal basis for $\mathcal{I}_{X,1}$, such that $E(\tilde{h}^2_l(X))=q(1-q), \forall  l\in \{1,2,\cdots, |\mathcal{X}|-1\}$, where $q=P(X\neq 0)$. Let $\tau= \{t:i_t=1\}$, then:
\begin{equation}
\tilde{e}_{\mathbf{i}}(X^n)=\sum_{\forall t\in \tau: l_{t}\in [1,|\mathcal{X}|-1]} c_{\mathbf{i},(l_{t})_{t\in\tau}}\prod_{t\in\tau}\tilde{h}_{l_t}(X_{{t}}),
\end{equation}
where $c_\mathbf{i},(l_{t})_{t\in\tau}\in \mathbb{R}$, and $(l_t)_{t\in\tau}$ is the sequence of $l_t$'s for $t\in \tau$.
\label{lem:one}
\end{Lemma}
\begin{proof}
 Follows from Definition \ref{Lem:tensor_dec}.
\end{proof}

\begin{Example}
  Let $\mathcal{X}=\{0,1\}$. Since $\mathcal{G}_{i_j}$'s, $j\in [1,n]$ take values from the set $\{\mathcal{I}_{X,1}, \gamma_{X,1}\}$, they are all one-dimensional. Let $\tilde{h}$ be defined as follows:
\begin{align}
\tilde{h}(X)=    \begin{cases}
      1-q, & \text{if } X=1, \\
      -q. & \text{if } X=0,
    \end{cases}
    \label{eq:basis} 
\end{align}
where $q\triangleq P(X=1)$. Then, the single element set $\{\tilde{h}(X)\}$ is a basis for $\mathcal{I}_{X,1}$. Hence, using the previous lemma:
\begin{equation}
\tilde{e}_{\mathbf{i}}(X^n)= c_{\mathbf{i}}\prod_{t:i_t=1}\tilde{h}(X_{{t}}),
\end{equation}
where $c_{\mathbf{i}}\in \mathbb{R}$.

\end{Example}

\subsection{Properties of the Real Decomposition}
 In the following we derive the main properties of the real decomposition proposed above. We are interested in the variance of $\tilde{e}_{\mathbf{i}}$'s. In the next proposition, we show that the $\tilde{e}_{\mathbf{i}}$'s are uncorrelated and we derive an expression for the variance of $\tilde{e}_{\mathbf{i}}$ using the notation in Lemma \ref{lem:one}.
\begin{Proposition}
\label{pr:partfun}
Define $\mathbf{P}_{\mathbf{i}}$ as the variance of $\tilde{e}_{\mathbf{i}}$. The following hold:
 \\ 1) $\mathbb{E}(\tilde{e}_{\mathbf{i}}\tilde{e}_{\mathbf{j}})=0, \mathbf{i}\neq \mathbf{j}$, in other words $\tilde{e}_{\mathbf{i}}$'s are uncorrelated.
 \\ 2)  $\mathbf{P}_{\mathbf{i}}=\mathbb{E}(\tilde{e}_{\mathbf{i}}^2)=\sum_{\forall t\in \tau: l_{t}\in [1,|\mathcal{X}|-1]}c_{\mathbf{i},(l_{t})_{t\in\tau}}^2\left(q(1-q)\right)^{w_H(\mathbf{i})}$, where $w_H$ is the Hamming weight function.. Particularly, if $\mathcal{X}=\{0,1\}$, then 
 $\mathbf{P}_{\mathbf{i}}=\mathbb{E}(\tilde{e}_{\mathbf{i}}^2)=c_{\mathbf{i}}^2\left(q(1-q)\right)^{w_H(\mathbf{i})}$. 
\end{Proposition}
\begin{proof}
 1) follows by direct calculation. 2) holds from the independence of $X_i$'s.
\end{proof}

In the next lemma we find a characterization of $\tilde{e}_{\mathbf{i}}, \mathbf{i}\in\{0,1\}^n$ for general $\tilde{e}$.

\begin{Proposition}
$\tilde{e}_{\mathbf{i}}=\mathbb{E}_{X^n|X_{\mathbf{i}}}(\tilde{e}|X_{\mathbf{i}})-\sum_{\mathbf{j}< \mathbf{i}} \tilde{e}_{\mathbf{j}}$ gives the unique orthogonal decomposition of $\tilde{e}$ into the Hilbert spaces $\mathcal{G}_{i_1}\otimes \mathcal{G}_{i_2}\cdots\otimes \mathcal{G}_{i_n}, \mathbf{i}\in \{0,1\}^n$.
\label{Lem:unique}
\end{Proposition}
\begin{proof}
 
The uniqueness of such a decomposition follows from the isomorphism relation stated in equation \eqref{eq:Hil_Dec2}. We prove that the $\tilde{e}_{\mathbf{i}}$ given in the lemma are indeed the decomposition into the components of the direct sum. Equivalently, we show that $1)$ $\tilde{e}=\sum_{\mathbf{i}}\tilde{e}_{\mathbf{i}}$, and $2)$ $\tilde{e}_{\mathbf{i}}\in \mathcal{G}_{i_1}\otimes \mathcal{G}_{i_2}\otimes\dotsb \otimes \mathcal{G}_{i_n}, \forall \mathbf{i}\in \{0,1\}^n$.

First we check the equality $\tilde{e}=\sum_{\mathbf{i}}\tilde{e}_{\mathbf{i}}$. Let $\mathbf{t}$ denote the n-length vector whose elements are all ones. We have:
\begin{align*}
\tilde{e}_{\mathbf{t}}=\mathbb{E}_{X^n|X_{\mathbf{t}}}(\tilde{e}|X_{\mathbf{t}})-\sum_{\mathbf{i}<\mathbf{t}}\tilde{e}_{\mathbf{i}}
\stackrel{(a)}{\Rightarrow} \tilde{e}_{\mathbf{t}}+\sum_{\mathbf{i}<\mathbf{t}}\tilde{e}_{\mathbf{i}}=\tilde{e}\stackrel{(b)}{\Rightarrow} \tilde{e}=\sum_{\mathbf{i}\in \{0,1\}^n}\tilde{e}_{\mathbf{i}},
\end{align*}
where in (a) we have used 1) $X_{\mathbf{t}}=X^n$ and 2) for any function $\tilde{f}$ of $X^n$, $\mathbb{E}_{X^n|X^n}(\tilde{f}|X^n)=\tilde{f}$, and (b) holds since $\mathbf{i}<\mathbf{t}\Leftrightarrow \mathbf{i}\neq \mathbf{t}$.
It remains to show that $\tilde{e}_{\mathbf{i}}\in \mathcal{G}_{i_1}\otimes \mathcal{G}_{i_2}\otimes\dotsb \otimes \mathcal{G}_{i_n}, \forall \mathbf{i}\in \{0,1\}^n$. The next proposition provides a means to verify this property.
\begin{Lemma}
 Fix $\mathbf{i}\in \{0,1\}^n$, define $\mathcal{A}_0\triangleq\{s|i_s=0\}$, and $\mathcal{A}_1\triangleq\{s|i_s=1\}$. $\tilde{f}$ is an element of $\mathcal{G}_{i_1}\otimes \mathcal{G}_{i_2}\otimes\dotsb \otimes \mathcal{G}_{i_n}$ if and only if $(1)$ it is constant in all $X_s$, $s\in \mathcal{A}_0$, and $(2)$ it has $0$ mean on all $X_s$, when $s\in \mathcal{A}_1$. 
\label{prop:belong1}\end{Lemma}
\begin{proof}
Please refer to the appendix.
\end{proof}
 Returning to the original problem, it is enough to show that $\tilde{e}_{\mathbf{i}}$'s satisfy the conditions in Lemma \ref{prop:belong1}.  We prove the stronger result presented in the next proposition.
\begin{Lemma}
\label{prop:belong2}
The following hold:
\nonumber\\1) $\forall{\mathbf{i}},\mathbb{E}_{X^n}(\tilde{e}_{\mathbf{i}})$=0.\\
 2) $\forall \mathbf{i}\leq \mathbf{k}$, we have $\mathbb{E}_{X^n|X_{\mathbf{k}}}(\tilde{e}_{\mathbf{i}}|X_{\mathbf{k}})=\tilde{e}_{\mathbf{i}}$.\\
 3) $\mathbb{E}_{X^n}(\tilde{e}_{\mathbf{i}}\tilde{e}_{\mathbf{k}})=0$, for $\mathbf{i}\neq \mathbf{k}$.\\
 4) $\forall \mathbf{k}\leq \mathbf{i}: \mathbb{E}_{X^n|X_{\mathbf{k}}}(\tilde{e}_{\mathbf{i}}|X_{\mathbf{k}})=0.$
 \end{Lemma}
\begin{proof}
 Please refer to the appendix.
\end{proof}
 The second condition in Lemma \ref{prop:belong2} is equivalent to condition (2) in Lemma \ref{prop:belong1}. The fourth condition in Lemma \ref{prop:belong2} is equivalent to condition (1) in Lemma \ref{prop:belong1}.
 Using Lemma \ref{prop:belong1} and \ref{prop:belong2}, we conclude that $\tilde{e}_{\mathbf{i}}\in \mathcal{G}_{i_1}\otimes \mathcal{G}_{i_2}\otimes\dotsb \otimes \mathcal{G}_{i_n}, \forall \mathbf{i}\in \{0,1\}^n$. This completes the proof of Proposition \ref{Lem:unique}.
\end{proof}
The following example clarifies the notation used in Proposition \ref{Lem:unique}.

\begin{Example}
 Consider the case where $n=2$. We have the following decomposition of $\mathcal{H}_{X,2}$:
\begin{equation}
 \mathcal{H}_{X,2}=(\mathcal{I}_{X,1} \otimes \mathcal{I}_{X,1}) \oplus (\mathcal{I}_{X,1} \otimes \gamma_{X,1})\oplus (\gamma_{X,1}\otimes\mathcal{I}_{X,1}) \oplus (\gamma_{X,1}\otimes\gamma_{X,1}).
 \label{eq:2dimdec}
 \end{equation}

Let $\tilde{e}(X_1,X_2)$ be an arbitrary function in $\mathcal{H}_{X,2}$. The unique decomposition of $\tilde{e}$ in the form given in \eqref{eq:2dimdec} is as follows:
\begin{align*}
 \tilde{e}&= \tilde{e}_{1,1}+\tilde{e}_{1,0}+\tilde{e}_{0,1}+\tilde{e}_{0,0},\\ 
 &\tilde{e}_{1,1}= \tilde{e}-\mathbb{E}_{X_2|X_1}(\tilde{e}|X_1)-\mathbb{E}_{X_1|X_2}(\tilde{e}|X_2)+\mathbb{E}_{X_1,X_2}(\tilde{e})\in\mathcal{I}_{X,1}\otimes\mathcal{I}_{X,1},
 \\& \tilde{e}_{1,0}=\mathbb{E}_{X_2|X_1}(\tilde{e}|X_1)-\mathbb{E}_{X_1,X_2}(\tilde{e})\in\mathcal{I}_{X,1}\times \gamma_{X,1},
\\& \tilde{e}_{0,1}= \mathbb{E}_{X_1|X_2}(\tilde{e}|X_2)-\mathbb{E}_{X_1,X_2}(\tilde{e})\in\gamma_{X,1}\otimes\mathcal{I}_{X,1},
\\&\tilde{e}_{0,0}=\mathbb{E}_{X_1,X_2}(\tilde{e})\in \gamma_{X,1}\otimes\gamma_{X,1}.
\end{align*}
It is straightforward to show that each of the $\tilde{e}_{i,j}$'s, $i,j\in \{0,1\}$, belong to their corresponding subspaces. For instance, $ \tilde{e}_{0,1}$ is constant in $X_1$, and is a $0$ mean function of $X_2$ (i.e. $\mathbb{E}_{X_2}\left(\tilde{e}_{0,1}(x_1,X_2)\right)=0, x_1\in \mathcal{X}$), so  $\tilde{e}_{0,1}\in \gamma_{X,1}\otimes\mathcal{I}_{X,1}$.
\end{Example}
Lastly, we derive an expression for $\mathbf{P}_{\mathbf{i}}$ using Proposition \ref{Lem:unique}:
\begin{Proposition}

For arbitrary $e:\mathcal{X}^n\to \{0,1\}$, let $\tilde{e}$ be the corresponding real-valued function, and let $\tilde{e}=\sum_{\mathbf{i}}\tilde{e}_{\mathbf{i}}$ be the decomposition in the form of Equation \eqref{eq:Hil_Dec2}. The variance of each component in the decomposition is given by the following recursive formula $\mathbf{P}_{\mathbf{i}} =\mathbb{E}_{X_{\mathbf{i}}}(\mathbb{E}_{X^n|X_{\mathbf{i}}}^2(\tilde{e}|X_{\mathbf{i}}))-\sum_{\mathbf{j}< \mathbf{i}}\mathbf{P}_{\mathbf{j}}, \forall \mathbf{i}\in \{0,1\}^n$, where $\mathbf{P}_{\underline{0}}\triangleq 0$. 
\label{Lem:power}
\end{Proposition}
\begin{proof}
 Please refer to the appendix.
\end{proof}
\begin{Corollary}
 For an arbitrary $e:\mathcal{X}^n\to \{0,1\}$ with corresponding real function $\tilde{e}$, and decomposition $\tilde{e}=\sum_{\mathbf{i}}\tilde{e}_{\mathbf{i}}$. Let the variance of $\tilde{e}$ be denoted by $\mathbf{P}$. Then, $\mathbf{P}=\sum_{\mathbf{i}}\mathbf{P}_{\mathbf{i}}$. 
 \label{Cor:power}
\end{Corollary}
The corollary is a special case of Proposition \ref{Lem:power}, where we have taken $\mathbf{i}$ to be the all ones vector.
So far we have found the decomposition of $\underline{\tilde{e}}$ into components $\underline{\tilde{e}_{\mathbf{i}}}$. 
 We can restate the claims made in the Introduction based on our new understanding of the effective-length: In various network communication problems, for an optimal encoding function $\underline{e}$,  the variance of $\underline{\tilde{e}}_{\mathbf{i}}$ is small when $\underline{\tilde{e}}_{\mathbf{i}}$ is a function of a large subset of the input. We prove this claim in the next three sections.
\section{Correlation Preservation in Arbitrary Encoders}\label{sec:corr}
%

Our goal is to bound the correlation preserving properties of general $n$-length encoding functions. As a first step, we derive bounds on the correlation between the outputs of two arbitrary Boolean functions (i.e. functions whose output is a binary scalar). The result that follows can be viewed as the solution to a more fundamental problem than what we discussed so far. We explain a summary of the setup here. Consider the two distributed agents shown in figure \ref{fig:agents}. 
For simplicity assume that the strings are produced based on a memoryless distribution. The strings are correlated with each other. The first agent is to make a binary decision based on its input. The second aims to guess the other's decision. 
We assume that the users have constraints on the effective-length of their decision functions. For instance, in the case of communication systems, the users have restrictions on their effective-lengths due to the rate-distortion requirements in the problem. For pedagogical reasons we present the results of this section in two parts. First, we consider binary input alphabets, and derive bounds on the probability of agreement of Boolean functions. Then, we extend these results to the case of non-binary input alphabets.
\subsection{Binary Input alphabets}

We proceed with presenting the main result of this section. Let $(X,Y)$ be a pair of binary DMS's. Consider two arbitrary Boolean functions $e:\{0,1\}^n\to \{0,1\}$ and $f:\{0,1\}^n\to \{0,1\}$. The following theorem provides an  upper-bound on the probability of equality between the functions $e(X^n)$ and $f(Y^n)$. 

\begin{Theorem}
 Let $\epsilon\triangleq P(X\neq Y)$, the following bound holds:
 \begin{equation*}
2\sqrt{\sum_{\mathbf{i}}\mathbf{P}_{\mathbf{i}}}\sqrt{\sum_{\mathbf{i}}\mathbf{Q}_{\mathbf{i}}}-2\sum_{\mathbf{i}}C_\mathbf{i}\mathbf{P}_{\mathbf{i}}^{\frac{1}{2}}\mathbf{Q}_{\mathbf{i}}^{\frac{1}{2}} 
\leq  P(e(X^n)\neq f(Y^n))\leq 1- 2\sqrt{\sum_{\mathbf{i}}\mathbf{P}_{\mathbf{i}}}\sqrt{\sum_{\mathbf{i}}\mathbf{Q}_{\mathbf{i}}}+2\sum_{\mathbf{i}}C_\mathbf{i}\mathbf{P}_{\mathbf{i}}^{\frac{1}{2}}\mathbf{Q}_{\mathbf{i}}^{\frac{1}{2}} 
,
\end{equation*}
 where $C_{\mathbf{i}}\triangleq  (1-2\epsilon)^{N_\mathbf{i}}$, $\mathbf{P}_{\mathbf{i}}$ is the variance of $\tilde{e}_{\mathbf{i}}$, and $\underline{\tilde{e}}$ is the real function corresponding to $\underline{e}$, and $\mathbf{Q}_{\mathbf{i}}$ is the variance of $\tilde{f}_{\mathbf{i}}$, and finally, $N_{\mathbf{i}}\triangleq w_H(\mathbf{i})$.
 \label{th:sec3}
\end{Theorem}
\begin{proof}
 Please refer to the appendix.
\end{proof}

\begin{Remark}
 The value $C_{\mathbf{i}}=(1-2\epsilon)^{N_\mathbf{i}}$ is decreasing with $N_{\mathbf{i}}$. So, in order to increase $ P(e(X^n)\neq f(Y^n))$, most of the variance $\mathbf{P}_{\mathbf{i}}$ should be distributed on $\tilde{e}_\mathbf{i}$ which have lower $N_{\mathbf{i}}$ (i.e. operate on smaller blocks). Particularly, the lower bound is minimized by setting
 \begin{align*}
 \mathbf{P}_{\mathbf{i}}=
\begin{cases}
 1  \qquad  & \mathbf{i}=\mathbf{i}_1,\\
0 \qquad & otherwise.
\end{cases}
\end{align*}
This recovers the result in \cite{ComInf2}.
 \end{Remark}
\begin{Remark}
 For fixed $\mathbf{P}_{\mathbf{i}}$, the lower-bound is minimized by taking $\tilde{e}$, and $\tilde{f}$ to be the same functions. 
\end{Remark}

\begin{Corollary}
We can simplify the bound in Theorem \ref{th:sec3} as follows:
 \begin{equation*}
2\sum_{\mathbf{i}}(1-C_\mathbf{i})\mathbf{P}_{\mathbf{i}}^{\frac{1}{2}}\mathbf{Q}_{\mathbf{i}}^{\frac{1}{2}} 
\leq  P(e(X^n)\neq f(Y^n))\leq 1-2\sum_{\mathbf{i}}(1-C_\mathbf{i})\mathbf{P}_{\mathbf{i}}^{\frac{1}{2}}\mathbf{Q}_{\mathbf{i}}^{\frac{1}{2}} 
\end{equation*}

\end{Corollary}
\begin{proof}
  \begin{align*}
 &\sigma \geq 2\sqrt{\sum_{\mathbf{i}}\mathbf{P}_{\mathbf{i}}}\sqrt{\sum_{\mathbf{i}}\mathbf{Q}_{\mathbf{i}}}-2\sum_{\mathbf{i}}C_\mathbf{i}\mathbf{P}_{\mathbf{i}}^{\frac{1}{2}}\mathbf{Q}_{\mathbf{i}}^{\frac{1}{2}} 
 \\&\stackrel{(a)}{\Rightarrow} \sigma
 \geq 2\sum_{\mathbf{i}}{{\mathbf{P}^{\frac{1}{2}}_{i}}}{\mathbf{Q}^{\frac{1}{2}}_{i}}-2\sum_{\mathbf{i}}C_\mathbf{i}\mathbf{P}_{\mathbf{i}}^{\frac{1}{2}}\mathbf{Q}_{\mathbf{i}}^{\frac{1}{2}}
 \\&\Rightarrow \sigma\geq2\sum_{\mathbf{i}}(1-C_{\mathbf{i}})\mathbf{P}_{\mathbf{i}}^{\frac{1}{2}}\mathbf{Q}_{\mathbf{i}}^{\frac{1}{2}}.
\end{align*}
In (a) we have used the Cauchy-Schwarz inequality. 
\end{proof}

\subsection{Arbitrary input alphabets}
So far, we have only considered Boolean functions with binary input alphabets. Next, we extend Theorem \ref{th:sec3}; and derive a new bound on the correlation between the outputs of Boolean functions with arbitrary (finite) input alphabets. Similar to the previous part, let $(X,Y)$ be a pair of DMS's with joint distribution $P_{X,Y}$. Assume that the alphabets $\mathcal{X}$ and $\mathcal{Y}$ are finite sets. 
 Consider two arbitrary Boolean functions $e:\mathcal{X}^n\to \{0,1\}$ and $f:\mathcal{Y}^n\to \{0,1\}$. We prove the following extension of Theorem \ref{th:sec3}.

\begin{Theorem}
 Let $\psi\triangleq \sup(E(h(X)g(Y))$, where the supremum is taken over all single-letter functions $h:\mathcal{X}\to \mathbb{R}$, and $g:\mathcal{Y}\to\mathbb{R}$ such that $h(X)$ and $g(Y)$ have unit variance and zero mean. 
  the following bound holds:
 \begin{equation*}
2\sqrt{\sum_{\mathbf{i}}\mathbf{P}_{\mathbf{i}}}\sqrt{\sum_{\mathbf{i}}\mathbf{Q}_{\mathbf{i}}}-2\sum_{\mathbf{i}}C_\mathbf{i}\mathbf{P}_{\mathbf{i}}^{\frac{1}{2}}\mathbf{Q}_{\mathbf{i}}^{\frac{1}{2}} 
\leq  P(e(X^n)\neq f(Y^n))\leq 1- 2\sqrt{\sum_{\mathbf{i}}\mathbf{P}_{\mathbf{i}}}\sqrt{\sum_{\mathbf{i}}\mathbf{Q}_{\mathbf{i}}}+2\sum_{\mathbf{i}}C_\mathbf{i}\mathbf{P}_{\mathbf{i}}^{\frac{1}{2}}\mathbf{Q}_{\mathbf{i}}^{\frac{1}{2}} 
,
\end{equation*}
 where 1) $C_{\mathbf{i}}\triangleq  \psi^{N_\mathbf{i}}$, 2) $\mathbf{P}_{\mathbf{i}}$ is the variance of $\tilde{e}_{\mathbf{i}}$, 3) $\underline{\tilde{e}}$ is the real function corresponding to $\underline{e}$, 4) $\mathbf{Q}_{\mathbf{i}}$ is the variance of $\tilde{f}_{\mathbf{i}}$, and 5) $N_{\mathbf{i}}\triangleq w_H(\mathbf{i})$.
 \label{th:sec3_non_bin}
\end{Theorem}
\begin{proof}
 Pleas refer to the Appendix.
\end{proof}
\begin{Remark}
 In Lemma \ref{Lem:init} it was shown that for binary random variables $X$ and $Y$, with $P(X\neq Y)=\epsilon$, we have $\psi=1-2\epsilon$. So, the bounds in Theorem \ref{th:sec3} and Theorem \ref{th:sec3_non_bin} are the same for binary inputs. 
\end{Remark}
\begin{Remark}
 The value of $\psi$ is in the interval $[0,1]$.  $\psi$ is equal to one if and only if $X=Y$. The proof is straightforward and follows from the Cauchy-Schwarz inequality.
\end{Remark}


%

\section{Correlation in Single Letter Coding Ensembles}\label{sec:SLCS}

In this section, we provide the definition of a coding ensemble, and identify three properties which are shared among many of the coding ensembles available for multi-terminal communications. We call the group of coding ensembles which share these properties the Single-letter Coding Ensembles (SLCE). These coding ensembles are the ones used in schemes such as Shannon's PtP source coding scheme, the Berger-Tung coding scheme for distributed source coding \cite{Markov}, the Zhang-Berger multiple-descriptions coding scheme \cite{ZB}, and the Cover-El Gamal-Salehi coding scheme \cite{CES} for transmission of correlated sources over the multiple-access channel. We use the results in the previous section to bound the correlation between the outputs of two encoding functions produced using a SLCE. The proof involves several steps. First, it is shown that SLCEs produce encoding functions which have most of their variance either on the single-letter components of their dependency spectrum or on the components with asymptotically large blocklengths. This along with Theorem \ref{th:sec3_non_bin} are used to prove that such schemes are inefficient in preserving correlation. 
\subsection{Single Letter Coding Ensembles}
Traditionally, in information theory, coding ensembles are used to show the existence of good encoding and decoding functions.
A coding ensemble is characterized by a probability distribution on the set of all encoding functions. 
We define the coding ensemble as follows:
\begin{Definition}
 A coding ensemble $\mathscr{S}$ is characterized by a probability measure $P_{\mathscr{S}}(\underline{e}_1,\underline{e}_2,\cdots,\underline{e}_t)$ on the set of encoding  functions $\underline{e}_k: \mathcal{X}^{n_k}\to \{0,1\}^{r_k}, k\in [1,t]$.
\end{Definition}
As an example, in Shannon's point-to-point coding ensemble, the probability of choosing any encoding and decoding function pair is determined by first using a single-letter distribution to assign probabilities to codebooks, and then using typicality rules to determine the encoding and decoding functions.  Whenever the choice of the coding ensemble is clear, we denote the distribution by $P_{\underline{E}}(\underline{e})$. The following defines the SLCEs:
\begin{Definition} \label{def:SLCE}
 The coding ensemble characterized by $P_{\mathscr{S}}$ is called an SLCE if the following constraints are satisfied. For an arbitrary $k\in [1,t]$, let $\underline{E}=\underline{E}_k$, then\footnote{Recall that the $i$th component of the vector of encoding functions $\underline{E}$ is denoted by $E_i$.}

 1) $ \exists \delta_X>0 \text{ such that }  \forall x^n, \exists B_n(x^n)\subset \mathcal{X}^n$ such that the following holds: 
\begin{align*}
&P_{X^n}(X^n\in B_n(x^n))\leq 2^{-n\delta_{X}},\text{ and }\forall \tilde{x}^n\notin B_n(x^n)
\\&(1-2^{-n\delta_X})P(\underline{E}(x^n))P(\underline{E}(\tilde{x}^n))<P(\underline{E}(x^n),\underline{E}(\tilde{x}^n))<(1+2^{-n\delta_X})P(\underline{E}(x^n))P(\underline{E}(\tilde{x}^n)) .
\end{align*}
 
 
 2) $ \forall \delta>0, \exists n \text{ such that }\mathbb{N}\ni \forall m>n, \forall x^m\in \{0,1\}^m , v \in \{0,1\}, \forall i\in [1,m]$:
 \\ $|P_{\mathscr{S}}(E_i(X^m)=v|X^m=x^m)-P_{\mathscr{S}}(E_i(X^m)=v|X_i=x_i)|<\delta$.

3) $\forall \pi \in S_n: P_{\underline{E}}(\underline{E})=P_{\underline{E}}(\underline{E}_\pi)$, where $\underline{E}_\pi(X^n)=\pi^{-1}(\underline{E}(\pi(X^n)))$, where $S_n$ is the symmetric group of length n.

\end{Definition}
\begin{Remark}
 The first condition can be interpreted as follows: for an arbitrary sequence $x^n$, let $B_n(x^n)$ be the set of sequences $\tilde{x}^n$, such that the set of encoding functions which map $x^n$ and $\tilde{x}^n$ to the same sequence has non-zero probability with respect to the measure $P_{\mathscr{S}} $ (e.g. in typicality encoding this requires $x^n$ and $\tilde{x}^n$ to be jointly typical based on some fixed distribution). Then, the condition requires that the probability of the set $B_n(x^n)$ goes to $0$ exponentially as $n\to \infty$. Furthermore, if $\tilde{x}^n\notin B_n(x^n)$, then it is mapped to a codeword which is chosen independent of $\underline{E}(x^n)$ (i.e. codewords are chosen pairwise independently). 
 \end{Remark}
 \begin{Remark}
The second condition can be interpreted as follows: the joint distribution of the input sequence and the output sequence of the encoding function averaged over all possible encoding functions approaches a product distribution in variational distance as $n\to \infty$.
\end{Remark}
\begin{Remark}
The explanation for the third condition is that the probability that a vector $x^n$ is mapped to $y^n$ depends only on their joint type and is equal to the probability that the permuted sequence $\pi(x^n)$ is mapped to $\pi(y^n)$. As an example, this condition holds in typicality encoding. 
\end{Remark}
In the next example we show that Shannon's coding ensemble for point-to-point source coding satisfies the above conditions.

\subsection{Examples of Single-letter Coding Ensembles}

In the following examples, we show that Shannon's point-to-point source coding strategy \cite{Shannon} and the Cover-El Gamal Salehi (CES) \cite{CES}  scheme for the transmission of correlated sources over the multiple access channel are SLCEs. 
\paragraph{Point-to-Point Source Coding}

 \begin{figure}[!t]
\centering
\includegraphics[width=0.7\columnwidth,draft=false]{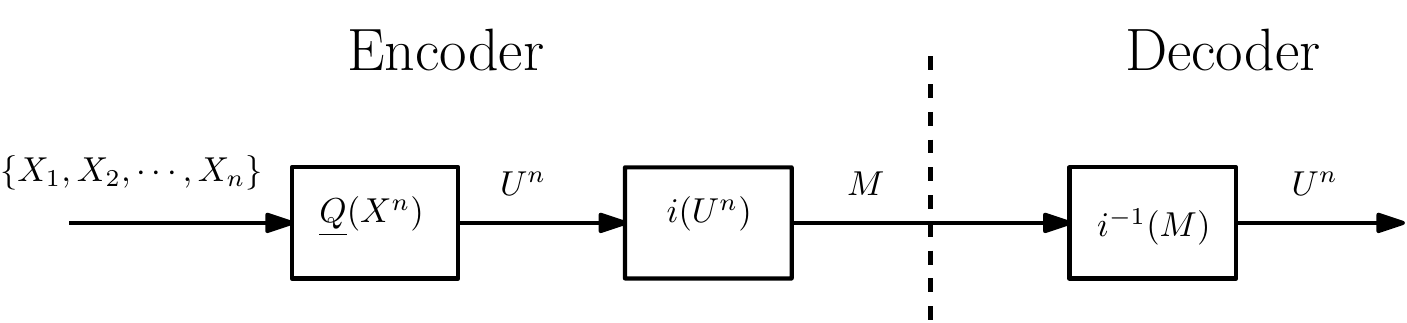}
\caption{Point-to-point source coding example}
\label{fig:PtPS}
\end{figure}
Consider the PtP source coding problem depicted in Figure \ref{fig:PtPS}. A discrete memoryless source $X$ is being fed to an encoder. The encoder utilizes the mapping $\underline{E}:\mathcal{X}^n\to \mathcal{U}^n$ to compress the source sequence. The image of $\underline{E}$ is indexed by the bijection $i:\textit{Im}(\underline{E})\to [1,|\textit{Im}(\underline{E})|]$. The index $M\triangleq{i(E(X^n))}$ is sent to the decoder. The decoder reconstructs the compressed sequence $U^n\triangleq i^{-1}(M)=\underline{E}(X^n)$.  The efficiency of the reconstruction is evaluated based on the separable distortion criteria $d_n:\mathcal{X}^n\times\mathcal{U}^n\to [0,\infty)$. The separability property means that $d_n(x^n,u^n)=\sum_{i\in[1,n]}d_1(x_i,u_i)$. We assume that the alphabets $\mathcal{X}$ and $\mathcal{U}$ are both binary. The rate of transmission is defined as $R\triangleq \frac{1}{n}\log{|\textit{Im}(\underline{E})|}$, and the average distortion is defined as $\frac{1}{n}\mathbb{E}(d_n(X^n, U^n))$. The goal is to choose $\underline{E}$ such that the rate-distortion tradeoff is optimized. Note that the choice of the bijection `$i$' is irrelevant to the performance of the system. The following Lemma gives the achievable RD region for this setup.  
\begin{Lemma}\cite{Shannon}\label{lem:PtPs}
 For the source $X$ and distortion criteria $d_1:\{0,1\}\times\{0,1\}\to [0,\infty)$, fix a conditional distribution  $p_{U|X}(u|x),x,u\in \{0,1\}$. 
 The rate-distortion pair $(R,D)=\left(r, \mathbb{E}_{X,U}(d_1(X,U)\right)) $ is achievable for all $r>I(U;X)$. 
\end{Lemma}
\begin{proof}
In order to verify the properties of the SLCE's in the coding ensemble used in this problem, we give an outline of the scheme. Fix $n\in \mathbb{N}$, and $\epsilon>0$. Define $P_U(u)=\mathbb{E}_{X}\{P_{U|X}(u|X))\}$. Proving achievability is equivalent to showing the existence of a suitable encoding function $\underline{E}(X^n)$. In \cite{Shannon}, a randomly generated encoding function is constructed with the aid of a set of vectors called the codebook, and an assignment rule called typicality encoding. We construct the codebook $\mathcal{C}$ as follows.  Let ${A}_{\epsilon}^n(U)\triangleq \{u^n\big||\frac{1}{n}w_{H}(u^n)-P_{U}(1)|<\epsilon\}$ be the set of $n$-length binary vectors which are $\epsilon$-typical with respect to $P_U$. Choose $\ceil{2^{nR}}$ vectors from ${A}_{\epsilon}^n(U)$ randomly and uniformly. Let $\mathcal{C}\subset {A}_{\epsilon}^n(U)$ be the set of these vectors. The encoder constructs the encoding function $\underline{E}(X^n)$ as follows.  For an arbitrary sequence $x^n\in \{0,1\}^n$, define $A_{\epsilon}^n(U|x^n)$ as the set of vectors in $\mathcal{C}$ which are jointly $\epsilon$-typical with $x^n$ based on $P_{U|X}$. The vector $\underline{E}(x^n)$ is chosen randomly and uniformly from $A_{\epsilon}^n(U|x^n)\cap \mathcal{C}$. The probabilistic choice of the codewords as well as the quantization, puts a distribution on the random function $\underline{E}$. It can be shown that as $n$ becomes larger codes produced based on this distribution $P(\underline{E})$ achieve the rate-distortion vector (R,D) with probability approaching one. 
 \end{proof}
 \label{Ex:PtPS}
\begin{Remark}
 It is well-known that in the above scheme, the codebook generation process could be altered in the following way: instead of choosing the codewords randomly and uniformly from the set of typical sequences ${A}_{\epsilon}^n(U)$, the encoder can produce each codeword independent of the others and with the distribution $P_{U^n}(u^n)=\Pi_{i\in[1,n]}P_{U}(u_i)$. However, the discussion that follows remains unchanged regardless of which of these codebook generation methods are used.  
\end{Remark}

\begin{Lemma}\label{lem:SLCE_def}
 The ensemble described in the proof of Lemma \ref{lem:PtPs} is a SLCE.
\end{Lemma}
\begin{proof}
 Please refer to the Appendix.
\end{proof}
\paragraph{Transmission of Correlated Sources over he Multiple Access Channel (CS-MAC)}

Consider the problem of the lossless transmission of the sources $U$ and $V$ over MAC depicted in Figure \ref{fig:CS-MACgen}. The largest known transmissible region for this problem is achieved using the CES scheme. The following lemma gives the transmissible region using the CES scheme in absence of common components (i.e. when there is no random variable $W$ such that (a) $H(W)>0$, (b) $W=f(U)=g(V)$.)
\begin{figure}[!h]
\centering
\includegraphics[width=0.6\columnwidth,draft=false]{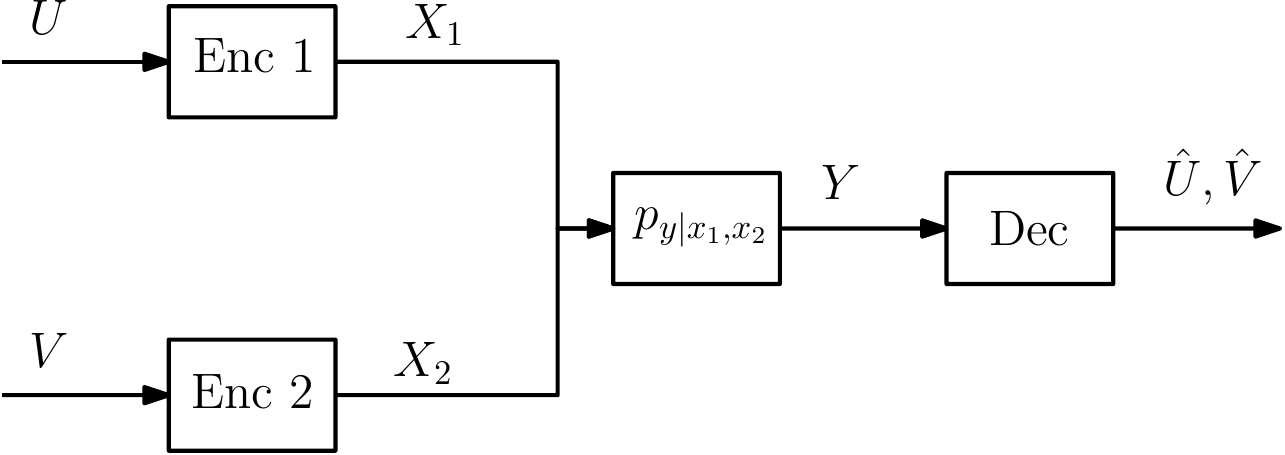}
\caption{Transmission of Sources over MAC}
\label{fig:CS-MACgen}
\end{figure}
\begin{Lemma}\cite{CES}
The sources $U$, and $V$ are transmissible over a a CS-MAC with channel input alphabets $\mathcal{X}_1$ and $\mathcal{X}_2$, and output alphabet $\mathcal{Y}$, and channel transition probability $p(y|x_1,x_2)$, if there exists a probability mass
function $p(x_1|u)p(x_2|v)$ such that:
\begin{align*}
 &H(U|V)< I(X_1;Y|X_2,V),\\
 &H(V|U)< I(X_2;Y|X_1,U),\\
 &H(U,V)< I(X_1,X_2;Y),
\end{align*}
where 
\begin{align*}
 p(u,v,x_1,x_2)=p(u,v)p(x_1|u)p(x_2|v).
\end{align*}
\end{Lemma}
Similar to the previous example achievability is proved by providing a coding ensemble which specifies
a probability distribution $P_{\mathscr{S}}(\mathbf{e}_1,\mathbf{e}_2)$ on the set of pairs of encoding functions at the two transmitters. 
The CES scheme generates the encoding functions independently (i.e. $P_{\mathscr{S}}(\mathbf{e}_1,\mathbf{e}_2)=P_{\mathscr{S}}(\mathbf{e}_1)P_{\mathscr{S}}(\mathbf{e}_2)$). Each of the encoding functions is generated by a method similar to the previous example. Hence, the marginals $P_{\mathscr{S}}(\mathbf{e}_i), i\in \{1,2\}$ each satisfy the conditions in Definition \ref{def:SLCE}. So, the coding ensemble is a SLCE. 

\begin{Remark}
 Similar arguments can be given to show that the Zhang-Berger coding scheme for multiple-descriptions coding \cite{ZB} uses a SLCE.
\end{Remark}

\subsection{Bounds on Output Correlation for SLCEs}
Our goal is to analyze the correlation preserving properties of SLCE's. For a randomly generated encoding function $\underline{E}= (E_1,E_2,\cdots,E_n)$, denote the decomposition of the real function corresponding to the $k$th element into the form in Equation \eqref{eq:Hil_Dec2} as $\tilde{E}_k= \sum_{\mathbf{i}}\tilde{E}_{k, \mathbf{i}}, k\in [1,n]$. Let $\mathbf{P}_{k,\mathbf{i}}$ be the variance of $\tilde{E}_{k, \mathbf{i}}$. The next theorem states one of the main results of this section.

\begin{Theorem}
 For any $k\in \mathbb{N}, m\in \mathbb{N}, \gamma>0$, $P_{\mathscr{S}}(\sum_{{\mathbf{i}}:N_{\mathbf{i}}\leq m, \mathbf{i}\neq \mathbf{i}_k} \mathbf{P}_{k,{\mathbf{i}}}\geq \gamma)\to 0$, as $n\to \infty$, where, $\mathbf{i}_k$ is the $k$th standard basis element.
 \label{th:sec4}
\end{Theorem}
\begin{Remark}
 Theorem \ref{th:sec4} shows that SLCE's distribute most of the variance of $\tilde{E}_k$ on $ \tilde{E}_{k, \mathbf{i}}$'s  which operate on large blocks. Hence, the encoders generated using such ensembles have high expected effective-lengths. This along with Theorem \ref{th:sec3} gives an upper bound on the correlation preserving properties of SLCE's.
\end{Remark}
\begin{proof}
The following proposition shows that the probability $P_{\mathscr{S}}(\sum_{{\mathbf{i}}:N_{\mathbf{i}}\leq m, \mathbf{i}\neq \mathbf{i}_k} \mathbf{P}_{k,{\mathbf{i}}}\geq \gamma)$ is independent of the index $k$. This is due to property 3) in the Definition \ref{def:SLCE} of SLCE's.
 \begin{Proposition}\label{prop:5.1}
 $P(\sum_{{\mathbf{i}}:N_{\mathbf{i}}\leq m, \mathbf{i}\neq 00\cdots01} \mathbf{P}_{k,{\mathbf{i}}}\geq \gamma)$ is constant as a function of $k$. 
\end{Proposition}
\begin{proof}
 Please refer to the Appendix. 
\end{proof}

Using the previous proposition, it is enough to show the theorem holds for $k=1$. For ease of notation we drop the subscript k for the rest of the proof and denote $\mathbf{P}_{1,\mathbf{i}}$ by $\mathbf{P}_{\mathbf{i}}$. 
By the Markov inequality, we have the following:
\begin{align}
 P_{\mathscr{S}}(\sum_{{\mathbf{i}}:N_{\mathbf{i}}\leq m, \mathbf{i}\neq \mathbf{i}_1} \mathbf{P}_{{\mathbf{i}}}\geq \gamma)\leq \frac{\sum_{{\mathbf{i}}:N_{\mathbf{i}}\leq m, \mathbf{i}\neq \mathbf{i}_1}\mathbb{E}_{\mathscr{S}}(\mathbf{P}_{\mathbf{i}})}{\gamma}.
 \label{eq:Markov}
\end{align}
So, we need to show that $\sum_{{\mathbf{i}}:N_{\mathbf{i}}\leq m, \mathbf{i}\neq \mathbf{i}_1}\mathbb{E}_{\mathscr{S}}(\mathbf{P}_{\mathbf{i}})$ goes to 0 for all fixed $m$. We first prove the following claim.

\begin{Claim}
Fix $\mathbf{i}$, the following holds:
\begin{align*}
 \mathbb{E}_{\tilde{E}, X_{\mathbf{i}}}(\mathbb{E}^2_{X^n|X_{\mathbf{i}}}(\tilde{E}|X_{\mathbf{i}}))= \mathbb{E}_{ X_{\mathbf{i}}}(\mathbb{E}^2_{\tilde{E}, X^n|X_{\mathbf{i}}}(\tilde{E}|X_{\mathbf{i}}))+O(e^{-n\delta_X}).
 \end{align*}

\label{claim:expo}
\end{Claim}
\begin{proof}
 Please refer to the Appendix.
\end{proof}
 Define $\bar{E}_{\mathbf{i}}= \mathbb{E}_{\tilde{E}}(\tilde{E}_{\mathbf{i}})= \mathbb{E}_{\tilde{E}|X_{\mathbf{i}}}(\tilde{E}|X_{\mathbf{i}})-\sum_{\mathbf{j}<\mathbf{i}}\bar{E}_{\mathbf{j}}$, and also define $\bar{P}_{\mathbf{i}}\triangleq Var(\bar{E}_{\mathbf{i}}) $. Using the above claim we have:
\begin{align}
 P_{\mathscr{S}}(\sum_{{\mathbf{i}}:N_{\mathbf{i}}\leq m, \mathbf{i}\neq \mathbf{i}_1} \mathbf{P}_{{\mathbf{i}}}\geq \gamma)\leq \frac{\sum_{{\mathbf{i}}:N_{\mathbf{i}}\leq m, \mathbf{i}\neq \mathbf{i}_1}\mathbb{E}_{\mathscr{S}}(\mathbf{P}_{\mathbf{i}})}{\gamma}
 \leq \frac{2^mO(e^{-n\delta_X})+     \sum_{{\mathbf{i}}:N_{\mathbf{i}}\leq m}\mathbb{E}_{\mathscr{S}}(\bar{P}_{\mathbf{i}})-\mathbb{E}_{\mathscr{S}} (\bar{P}_{\mathbf{i}_1})}{\gamma}.
 \label{eq:follow}
\end{align}

Using the arguments from the proof of Lemma \ref{prop:belong2}, we can see that the properties stated in that Proposition hold for $\bar{E}_{\mathbf{i}}$ as well.  By the same results as in Lemma \ref{Lem:power} and Corollary \ref{Cor:power}, we have that $\sum_{\mathbf{i}\in\{0,1\}^n}\bar{P}_{\mathbf{i}}=\bar{P}_{\mathbf{\underline{1}}}$. Following the calculations in \eqref{eq:follow}:
\begin{align*}
 P_{\mathscr{S}}(\sum_{{\mathbf{i}}:N_{\mathbf{i}}\leq m, \mathbf{i}\neq \mathbf{i}_1} \mathbf{P}_{{\mathbf{i}}}\geq \gamma)&
 \leq \frac{2^mO(e^{-n\delta_X})+     \sum_{{\mathbf{i}}:N_{\mathbf{i}}\leq m}\mathbb{E}_{\mathscr{S}}(\bar{P}_{\mathbf{i}})-\mathbb{E}_{\mathscr{S}} (\bar{P}_{\mathbf{i}_1})}{\gamma}\\
& 
 \leq \frac{2^mO(e^{-n\delta_X})+     \sum_{{\mathbf{i}}\in \{0,1\}^n}\mathbb{E}_{\mathscr{S}}(\bar{P}_{\mathbf{i}}) -\mathbb{E}_{\mathscr{S}} (\bar{P}_{\mathbf{i}_1})}{\gamma} \\
 &=
 \frac{2^mO(e^{-n\delta_X})+    \mathbb{E}_{\mathscr{S}}( \sum_{{\mathbf{i}}\in \{0,1\}^n}\bar{P}_{\mathbf{i}}) -\mathbb{E}_{\mathscr{S}} (\bar{P}_{\mathbf{i}_1})}{\gamma} \\
 &= \frac{2^mO(e^{-n\delta_X})+    \mathbb{E}_{X^n}\left(\mathbb{E}^2_{\tilde{E}|X^n}(\tilde{E}(X^n)|X^n)\right) -\mathbb{E}_{\mathscr{S}} (\bar{P}_{\mathbf{i}_1})}{\gamma}\\
 &\leq \frac{2^mO(e^{-n\delta_X})+    \mathbb{E}_{\mathscr{S}}(\bar{P}_{\mathbf{i}_1})+O(\epsilon) -\mathbb{E}_{\mathscr{S}} (\bar{P}_{\mathbf{i}_1})}{\gamma}\\
 &= \frac{2^mO(e^{-n\delta_X})+O(\epsilon) }{\gamma},
\end{align*}
where in the last inequality we have used the second property in Definition \ref{def:SLCE}. The last line goes to 0 as $n\to\infty$. This completes the proof.
\end{proof}

Theorem \ref{th:sec4} shows that SLCE's distribute most of the variance of $\tilde{E}_k$ on $ \tilde{E}_{k, \mathbf{i}}$'s  which operate on asymptotically large blocks of the input, and on the single-letter component $\tilde{E}_{k, \mathbf{i}_k}$. Hence, the encoders generated using such ensembles have high expected variance for decomposition elements with large effective-lengths. This along with Theorem \ref{th:sec3} gives an upper bound on the correlation preserving properties of SLCE's. The following theorem states this upper bound:

\begin{Theorem}
Let $(X,Y)$ be a pair of DMS's, with $P(X=Y)=1-\epsilon$. Also, assume that the pair of BBE's $\underline{E},\underline{F}$ are produced using SLCE's. Define $E\triangleq E_{1}$, and $F\triangleq F_1$.
Then,
\begin{align*}
 \forall{\delta>0}:P_{\mathscr{S}}\left(P_{X^n,Y^n}\left(E(X^n)\neq F(Y^n)\right)> 2\mathbf{P}^{\frac{1}{2}}\mathbf{Q}^{\frac{1}{2}}- 2(1-2\epsilon)\mathbf{P}_{\mathbf{i}_1}^{\frac{1}{2}}\mathbf{Q}_{\mathbf{i}_1}^{\frac{1}{2}}-\delta\right)\to 1,
\end{align*}
as $n\to \infty$. Where $\mathbf{P}_{\mathbf{i}}\triangleq Var(\tilde{E}_{\mathbf{i}})$, $\mathbf{Q}_{\mathbf{i}}\triangleq Var(\tilde{F}_{\mathbf{i}})$, $\mathbf{P}\triangleq Var(\tilde{E})$, and $ \mathbf{Q}\triangleq Var(\tilde{F})$.
    \label{th:main}\end{Theorem}
\begin{Remark}
 Note that in this theorem we consider a pair of BBEs produced using SLCEs. The bound is presented as function of the dependency spectra of the two BBEs. The two SLCEs can have arbitrary correlation. As an example, $E$ and $F$ can be taken to be either independent or exactly equal to each other. 
\end{Remark}
\begin{proof}
  From Theorem \ref{th:sec3}, we have:
  \begin{align*}
     \mathbf{P}^{\frac{1}{2}}\mathbf{Q}^{\frac{1}{2}} -2\sum_{\mathbf{i}}C_{\mathbf{i}}\mathbf{P}^{\frac{1}{2}}_{\mathbf{i}}\mathbf{Q}^{\frac{1}{2}}_{\mathbf{i}}\leq  P(E(X^n)\neq F(Y^n)).
\end{align*}
From Theorem \ref{th:sec4} we have:
\begin{align}
& \forall  m\in \mathbb{N}, \gamma>0, P_{\mathscr{S}}(\sum_{{\mathbf{i}}:N_{\mathbf{i}}\leq m, \mathbf{i}\neq \mathbf{i}_1} \mathbf{P}_{{\mathbf{i}}}< \gamma)\to 1, \qquad P_{\mathscr{S}}(\sum_{{\mathbf{i}}:N_{\mathbf{i}}\leq m, \mathbf{i}\neq \mathbf{i}_1} \mathbf{Q}_{{\mathbf{i}}}< \gamma)\to 1.
\label{eq:bound}
\end{align}
Note that:
\begin{align}
 \sum_{{\mathbf{i}}:N_{\mathbf{i}}\leq m, \mathbf{i}\neq \mathbf{i}_1} \mathbf{P}_{{\mathbf{i}}}< \gamma ,
 \sum_{{\mathbf{i}}:N_{\mathbf{i}}\leq m, \mathbf{i}\neq \mathbf{i}_1} \mathbf{Q}_{{\mathbf{i}}}< \gamma
   \Rightarrow
 \sum_{\mathbf{i}}C_{\mathbf{i}}\mathbf{P}^{\frac{1}{2}}_{\mathbf{i}}\mathbf{Q}^{\frac{1}{2}}_{\mathbf{i}} > (1-2\epsilon)(\mathbf{P}_{\mathbf{i}_1}+\gamma)^{\frac{1}{2}}(\mathbf{Q}_{\mathbf{i}_1}+\gamma)^{\frac{1}{2}}+ (1-2\epsilon)^m \mathbf{P}^{\frac{1}{2}}\mathbf{Q}^{\frac{1}{2}},
 \label{eq:bound2}
\end{align}
which converges to $(1-2\epsilon)\mathbf{P}_{\mathbf{i}_1}^{\frac{1}{2}}\mathbf{Q}_{\mathbf{i}_1}^{\frac{1}{2}}+(1-2\epsilon)^m\mathbf{P}^{\frac{1}{2}}\mathbf{Q}^{\frac{1}{2}}$ as $\gamma\to 0$. Also $C_{\mathbf{i}}$ is decreasing in $N_{\mathbf{i}}$ and goes to 0 as $N_{\mathbf{i}}\to \infty$. Choose $\gamma$ small enough and $m$ large enough such that $(1-2\epsilon)(\mathbf{P}_{\mathbf{i}_1}+\gamma)^{\frac{1}{2}}(\mathbf{Q}_{\mathbf{i}_1}+\gamma)^{\frac{1}{2}}+(1-2\epsilon)^m \mathbf{P}^{\frac{1}{2}}\mathbf{Q}^{\frac{1}{2}}-(1-2\epsilon)\mathbf{P}_{\mathbf{i}_1}^{\frac{1}{2}}\mathbf{Q}_{\mathbf{i}_1}^{\frac{1}{2}}< \delta $. Then Equations \eqref{eq:bound} and \eqref{eq:bound2} gives 
\[
P_{\mathscr{S}}\left(P_{X^n,Y^n}\left(E(X^n)\neq F(Y^n)\right)<2\mathbf{P}^{\frac{1}{2}}\mathbf{Q}^{\frac{1}{2}}- 2(1-2\epsilon)\mathbf{P}_{\mathbf{i}_1}^{\frac{1}{2}}\mathbf{Q}_{\mathbf{i}_1}^{\frac{1}{2}}-\delta\right)\to 0.\]
This is equivalent to the statement of the theorem.
\end{proof}
\begin{Remark}
 The previous theorem gives a bound on the correlation preserving properties on SLCE's. The theorem shows that in order to increase correlation in these schemes the encoder needs to put more variance on the element $\tilde{E}_{k,\mathbf{i}_k}, k\in [1,n]$. This would require more correlation between the input and output of the encoder, which itself would require more rate. As an example consider the extreme case where $Var(\tilde{E}_k)=Var(\tilde{E}_{k,\mathbf{i}_k})$, which requires ${E}_k(X^n)=X_k$. This means that in order to achieve maximum correlation, the encoder must use uncoded transmission. 
\end{Remark}
\begin{Remark}
 In the case when $X=Y$, there is common-information \cite{ComInf1} available at the encoders. If the encoders use the same encoding function $E$, their outputs would be equal. Whereas from theorem \ref{th:main}, for any non-zero $\epsilon$, the output correlation is bounded away from 0 (except when doing uncoded transmission). So, the correlation between the outputs of SLCE's is discontinuous as a function of $\epsilon$.
\end{Remark}

\section{Multi-terminal Communication Examples}\label{sec:ex}
In this section, we provide two examples of multi-terminal communication problems where SLCEs have suboptimal performance. We use the discontinuity mentioned in the previous section to show the sub-optimality of SLCEs. 

\paragraph{Transmission of Correlated Sources over the Interference Channel}
Consider the problem of transmission of correlated sources over the interference channel (CS-IC) described in \cite{ICcorr}. We examine the specific CS-IC setup shown in Figure \ref{fig:ICcorr}. We are restricting our attention to bandwidth expansion factor equal to one.
 Here, the sources $X$ and $Y$ are Bernoulli random variables with parameters $\alpha_X$ and $\alpha_Y$, and $Z$ is a $q$-ary random variable with distribution $P_Z$. $X$ and $Z$ are independent. $Y$ and $Z$ are also independent. Finally, $X$ and $Y$ are correlated, and $P(X\neq Y)=\epsilon$. The random variable $N_{\delta}$ is Bernoulli with parameter $\delta$. Decoder one reconstructs $X$ and decoder two reconstructs $Z$ losslessly.
 The first transmitter transmits the binary input $X_1$, and the second transmitter transmits the pair of inputs $(X_{21}, X_{22})$, where $X_{21}$ is $q$-ary and $X_{22}$ is binary. Receiver 1 receives $Y_1=X_1\oplus_2N_{\delta}$, and receiver 2 receives $Y_2$ which is given below:
\begin{align}\label{eq:y2}
 Y_2= 
 \begin{cases}
 X_{21}, \qquad \qquad &\text{if }X_{22}=X_1,\\
 e,     & otherwise.
\end{cases}
\end{align}
\begin{figure}[!t]
\centering
\includegraphics[height=1.2in, draft=false]{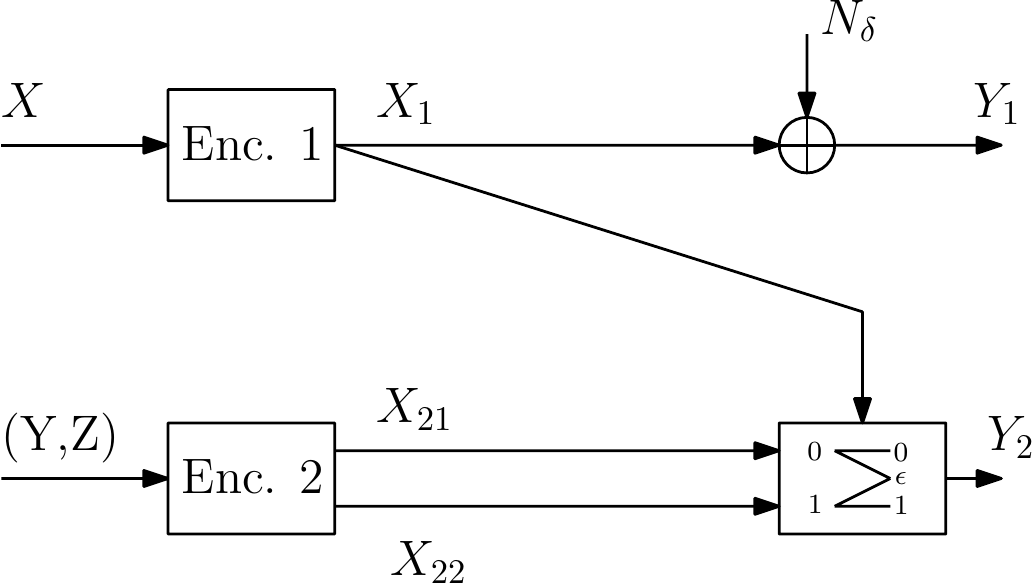}
 \caption{An CS-IC example where SLCEs are suboptiomal.}
\label{fig:ICcorr}
\end{figure}
So, the second channel outputs $X_{21}$ noiselessly if the second encoder `guesses' the first encoder's output correctly (i.e. $X_{22}=X_1$), otherwise an erasure is produced. The following proposition gives a set of sufficient conditions for the transmission of correlated sources over this interference channel:
\begin{Proposition}\label{prop:ICcorrach}
 The sources $X$ and $Z$ are transmissible if there exist $\epsilon, \gamma,d >0$, and $n\in \mathbb{N}$ such that:
 \begin{align*}
 & H(X)\leq (1-h_b(\delta))\left(1- \frac{h_b(\gamma+d)}{1-h_b(\delta)}\right)+O\left(\frac{1+ \sqrt{nV+k\mathcal{V}(d)}Q^{-1}(\gamma)}{\sqrt{n}}\right),\\
 &  H(Z)\leq \left((1-\epsilon)^k\log{q}\right)\left(1- \frac{h_b(\gamma+d)}{1-h_b(\delta)}\right),
\end{align*}
 where $h_b(\cdot)$ is the binary entropy function, $V=\delta(1-\delta)log_2{(\frac{1-\delta}{\delta})}$ is the channel dispersion, and $\mathcal{V}(d)$ is the rate-dispersion function as in \cite{KostinaSC}, and $Q(.)$ is the Gaussian complementary cumulative distribution function.
  
 For a fixed $n$, $\epsilon$ and $\gamma$, we denote the set of pairs $(H(X), H(Z))$ which satisfy the bounds by $S(n,\epsilon,\gamma)$.
 \end{Proposition}
 \begin{proof}
First we provide an outline of the coding strategy. Fix $n,m\in \mathbb{N}, d, \gamma\in \mathbb{R}$, where $n\ll m$. Let $k=n\left(1-\frac{h_b(d+\gamma)}{1-h_b(\delta)}\right)^{-1}$. The encoders send $km$ bits of the compressed input at each block of transmission.
The first encoder transmits its source in two steps. First, it uses a fixed blocklength source-channel code \cite{KostinaSC} with parameters $(k,n,d,\gamma)$ . The code maps $k$-length blocks of the source to $n$-length blocks of the channel input, and the average distortion resulting from the code is less than $d+\gamma$. In this step, the encoder transmits the source in $m$ blocks of length $k$. A total of $nm$ channel uses are needed (note that n<k). In the second step, the encoder uses a large blocklength code to correct the errors in the previous step. 
The code has rate close to $\frac{h_b(\gamma+d)}{1-h_b(\delta)}$, and its input length is equal to $km$. 

The second encoder only transmits messages in the first step of transmission. It uses the same fixed blocklength code as the first encoder and the source sequence $Y^k$ to estimate the outcome of the first encoder. It sends this estimate of the first encoder's output on $X^n_{22}$. Since $P(X^k=Y^k)=(1-\epsilon)^k$, we conclude that  $X^k_1$ and $X^k_{22}$ are equal at least with probability $(1-\epsilon)^k$. The encoder sends the source $Z$ using $X_{21}$ over the resulting $q$-ary erasure channel which has probability of erasure at most $(1-\epsilon)^k$. The following provides a detailed descriptions of the coding strategy:
\\\textbf{Codebook Generation:} Fix $n, \epsilon, d$. Let $k=n\left(1-\frac{h_b(d+\gamma)}{1-h_b(\delta)}\right)^{-1}$. 
Let $C_k$ be the optimal source-channel code with parameters $(k,n,d,\gamma)$ for the point-to-point transmission of a binary source over the binary symmetric channel, as described in \cite{KostinaSC}. The code transmits $k$-length blocks of the source using $n$-length blocks of the channel input; and guarantees that the resulting distortion at each block is less than $d$ with probability $(1-\epsilon)$ (i.e $P(d_H(X^n,\hat{X}^n)>d)\leq \gamma$, where $\hat{X}$ is the reconstruction of the binary source $X$ at the decoder). In \cite{KostinaSC}, it is shown that the parameters of the code satisfy:
\begin{align*}
 n(1-h_b(\delta))-&k\left(H(X)-h_b(\alpha_X\ast d)\right)
 \\&= \sqrt{nV+k\mathcal{V}(d)}Q^{-1}(\gamma)+O({log(n)}).
\end{align*}
Since $P(d_H(X^n,\hat{X}^n)>d)\leq \gamma$, it is straightforward to show that the average distortion is less than or equal to $\gamma+d$. Also, construct a family of good channel codes $C'_m, m\in \mathbb{N}$ for the binary symmetric channel with rate $R_m=1-h_b(\delta)-\lambda_m$, where $\lambda_m\to 0$ as $m\to \infty$. Next, construct a family of good channel codes $C''_m, m\in \mathbb{N}$ for the $q$-ary erasure channel with rate $R_m=(1-\epsilon)^klog(q)-\lambda_m$. Finally, randomly and uniformly bin the space of binary vectors of length $kn$ with rate $R'=h_n(d+\gamma)$. More precisely, generate a binning function $B:\{0,1\}^{km}\to \{0,1\}^{kmR'}$, by mapping any vector $\mathbf{i}$ to a value chosen uniformly from $\{0,1\}^{kmR'}$.
\\\textbf{Encoding:} Fix $m$. At each block the encoders transmit $km$ symbols of the source input. Let the source sequences be denoted by $X(1:k,1:m), Y(1:k,1:m), Z(1:k,1:m)$, where we have broken the source vectors into $m$ blocks of length $k$. In this notation $X(i,j)$ is the $i$th element of the $j$th block, and $X(1:k,j), j\in [1,m]$ is the $j$th block.
\\\textbf{Step 1:} Encoder 1 uses the code $C_k$ to transmit each of the blocks $X(1:k,i), i\in [1,m]$ to the decoder. The second encoder finds the output of the code $C_k$ when $Y(1:k,i)$ is fed to the code, and transmits the output vector on $X_{22}(1:n,i)$. The encoder uses an interleaving method similar to the one in \cite{FinLen} to transmit $Z$. For the sequence $Z(1:k,1:m)$, it finds the output of $C''_{km}$ for this input and transmits it on $X_{21}(1:n,1:m)$.
\\\textbf{Step 2:} The first encoder transmits $B(X(1:k,1:m))$ to the decoder losslessly using $C'_{kmR'}$.
\\\textbf{Decoding:} In the first step, the first decoder reconstructs $X(1:k,1:m)$ with average distortion at most $\gamma+d$. In the second step, using the bin number $B(X(1:k,1:m))$ it can losslessly reconstruct the source, since $C'_{kmR'}$ is a good channel code. Decoder 2 also recovers $Z(1:k,1:m)$ losslessly using $Y_2(1:k,1:m)$ since $C''_{km}$ is a good channel code. 

The conditions for successful transmission is given as follows:
\begin{align*}
 n(1-h_b(\delta))-&k\left(H(X)-h_b(\alpha_X\ast d)\right)
 \\&\geq \sqrt{nV+k\mathcal{V}(d)}Q^{-1}(\gamma)+O({log(n)}),\\
 &n(1-\epsilon)^klog(q)\geq kH(Z).
\end{align*}
Simplifying these conditions by replacing $k=n\left(1-\frac{h_b(d+\gamma)}{1-h_b(\delta)}\right)^{-1}$ proves the proposition.
\end{proof}
 
 The bound provided in Proposition \ref{prop:ICcorrach} is not calculable without the exact characterization of the $O(\frac{log(n)}{n})$ term. However, we use this bound to prove the sub-optimality of SLCEs. First, we argue that the transmissible region is `continuous' as a function of $\epsilon$. Note that for $\epsilon=0$, sources with parameters $(H(X),H(Z))=(1-h_b(\delta), \log{q})$ are transmissible. The region in Proposition \ref{prop:ICcorrach} is continuous in the sense that as $\epsilon$ approaches 0, the pairs $(H(X),H(Z))$ in the neighborhood of $(1-h_b(\delta), \log{q})$ satisfy the bounds given in the proposition (i.e. the corresponding sources are transmissible).
\begin{Proposition}
For all $\lambda>0$, there exist $\epsilon_0, \gamma_0>0$, and $n_0\in \mathbb{N}$ such that: 
 \begin{align*}
 \forall \epsilon<\epsilon_0: (1-h_b(\delta)-\lambda, \log{q}-\lambda)\in S(n_0,\epsilon, \gamma_0).
\end{align*}
\end{Proposition}
\begin{proof}
 Follows directly from Proposition \ref{prop:ICcorrach}.
\end{proof}
For an arbitrary encoding scheme operating on blocks of length $n$, let the encoding functions be as follows: $X_1^n=\underline{e}_1(X^n)$, and $(X_{21}^n, X_{22}^n)=(\underline{e}_{21}(Y^n,Z^n), \underline{e}_{22}(Y^n,Z^n))$. The following lemma gives an outer bound on $H(Z)$ as a function of the correlation between the outputs of $\underline{e}_1$ and $\underline{e}_{21}$.
\begin{Lemma}
 For a coding scheme with encoding functions $\underline{e}_1(X^n), \underline{e}_{21}(Y^n,Z^n), \underline{e}_{22}(Y^n,Z^n)$, the following holds:
 \begin{align}
 \label{eq:Hz}
 H(Z)\leq \frac{1}{n}\sum_{i=1}^n P(e_{1,i}(X^n)=e_{22,i}(Y^n,Z^n))+1.
\end{align}
\label{lem:corrIC}
\end{Lemma}
\begin{proof}
 Since $Z^n$ is reconstructed losslessly at the decoder, by Fano's inequality the following holds:
 \begin{align*}
 &H(Z^n)\approx I(Y_2^n;Z^n)\stackrel{(a)}{=} I(E^n,Y_2^n;Z^n)= I(E^n;Z^n)+I(Y_2^n;Z^n|E^n)\\
 &\stackrel{(b)}\leq H(E^n)+\sum_{i=1}^n P(e_{1,i}(X^n)=e_{22,i}(Y^n,Z^n))\log{q}
 \\&\stackrel{(c)}{\leq} n+\sum_{i=1}^n P(e_{1,i}(X^n)=e_{22,i}(Y^n,Z^n))\log{q},
\end{align*}
where in (a) we have defined $E^n$ as the indicator function of the event that $Y_2=e$, in (b) we have used Equation \ref{eq:y2}
and in (c) we have used the fact that $E^n$ is binary.
 \end{proof}
Using Theorem \ref{th:main}, we show that if the encoding functions are generated using SLCEs, $P(e_{1,i}(X^n)=e_{22,i}(Y^n,Z^n))$ is discontinuous in $\epsilon$. The next proposition shows that SLCEs are sub-optimal:
\begin{Proposition}
 There exists $\lambda>0$, and $q\in \mathbb{N}$, such that sources with $(H(X),H(Z))=(1-h_b(\delta)-\lambda, \log{q}-\lambda)$ are not transmissible using SLCEs. 
\end{Proposition}
\begin{proof}
 Let $X_1^n=\underline{E}_{1}(X^n)$, and $X_{22^n}=\underline{E}_{22,z^n}(Y^n), z^n\in \{0,1\}^n$ be the encoding functions used in the two encoders to generate $X_1$ and $X_{22}$. If $H(Z)\approx \log(q)$, from \eqref{eq:Hz}, we must have $P(E_{1,j}(X^n)=E_{22,z^n,j}(Y^n))\approx 1 $ for almost all of the indices $j\in[1,n]$. From Theorem \ref{th:main}, this requires $\mathbf{P}_{j,\mathbf{i}_{j}}\approx1$, which requires uncoded transmission (i.e. $X_1^n=\underline{E}(X^n)\approx X^n$). However, uncoded transmission contradicts the lossless reconstruction of the source at the first decoder. 
  \end{proof}
The proof is not restricted to any particular scheme, rather it shows that any SLCE would have sub-optimal performance.

\paragraph{Transmission of Correlated Sources over the Multiple Access Channel (CS-MAC)}
\begin{figure}[!t]
\centering
\includegraphics[height=1.2in, draft=false]{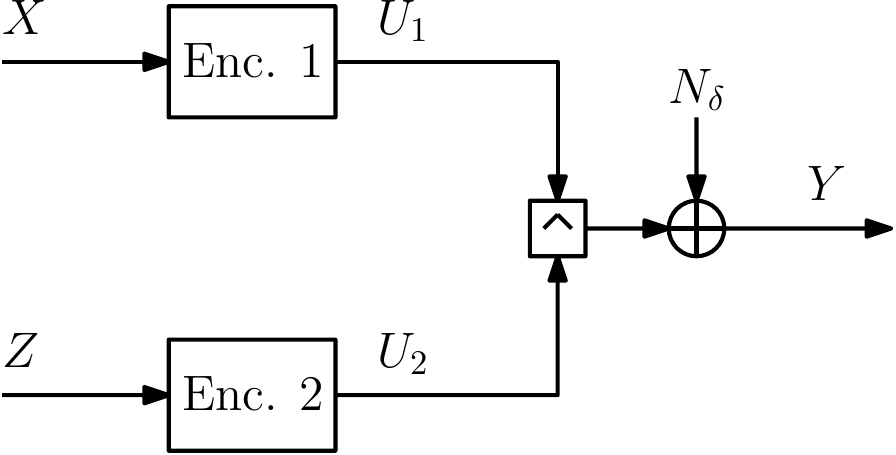}
 \caption{A CS-MAC example where SLCEs are suboptiomal.}
\label{fig:MACcorr}
\end{figure}

 We examine the CS-MAC setup shown in Figure \ref{fig:MACcorr}. Again, we restrict our attention to bandwidth expansion factor equal to one. 
 Here, the source $X$ is a $q$-ary source. The source $Z$ is defined as $Z=X\oplus_q N_{\epsilon}$, where $N_{\epsilon}$ is a $q$-ary random variable with 
\begin{align*}
 P(N_{\epsilon}=i)=
 \begin{cases}
 1-\epsilon, \qquad &\text{if } i=0,\\
 \frac{\epsilon}{q-1} \qquad &\text{if } i\in \{1,2,\cdots,q-1\},
\end{cases}
\end{align*}
and
 \begin{align*}
 P(N_{\delta}=i)=
 \begin{cases}
 1-\delta, \qquad &\text{if } i=0,\\
 \frac{\delta}{q-1} \qquad &\text{if } i\in \{1,2,\cdots,q-1\},
 \end{cases}
\end{align*}
 The output is:
\begin{align*}
Y=U_1\wedge U_2\oplus_q N_\delta=
 \begin{cases}
 U_2\oplus_qN_\delta, \qquad &\text{if } U_1=U_2,\\
 N_\delta \qquad  &\text{if } U_1\neq U_2,
\end{cases}
\end{align*}
 where $U_1$ and $U_2$ are the outputs of Encoder 1 and  Encoder 2, respectively. The goal is to transmit both sources $X$ and $Z$ losslessly to the decoder.
 
 In this setup, there are two strategies available to the encoders. The first strategy is for both encoders to transmit the sources simultaneously. In this case, the encoders must have equal outputs. Otherwise, the decoder receives the noise $N_{\delta}$. So, in this strategy, the encoders must `guess' each other's outputs.  The second strategy is to make a binary symmetric channel with noise $\delta$ for one of the encoders, while the other encoder does not transmit any messages. For example, in order to create such a channel for Encoder 1,
 encoder two transmits a constant sequence $U_2^n=(j,j,\cdots,j), j\in [1,q-1]$. Then, Encoder 2 can transmit a binary codeword using alphabet $\{0,j\}$. The rates of transmission for this strategy is:
 \begin{align*}
& R_{s,1}=\max_{p(U_1)} I(U_1;Y)= h\left(\frac{1}{2}\left(1-\frac{(q-2)\delta}{q-1}\right),\frac{1}{2}\left(1-\frac{(q-2)\delta}{q-1}\right), \frac{\delta}{q-1},\cdots, \frac{\delta}{q-1}\right)
- h\left(1-\delta,  \frac{\delta}{q-1}, \cdots, \frac{\delta}{q-1}\right),\\
 &R_{s,2}=0.
\end{align*}
 The following Proposition gives a condition under which the sources are transmissible:
 \begin{Proposition}
 There exists positive reals $\lambda_\epsilon, \epsilon\in (0,\frac{1}{2}]$, with $\lim_{\epsilon\to 0}\lambda_{\epsilon}=0$, such that the sources $X$ and $Y$ are transmissible if the following condition is satisfied:
 \begin{align*}
 H(X)\leq \log{q}-H(N_\delta)-\lambda_{\epsilon}.
\end{align*}
\end{Proposition}
\begin{proof}
The ideas in this proof are similar to the ones in Proposition \ref{prop:ICcorrach}. We provide an outline of the proof here. There are two steps for the transmission of the sources. First, the first strategy described above is used to transmit at a rate close to $\log{q}-H(N_\delta)$. In this step, the encoders use a finite blocklength code to maximize their probability of agreement. In the second step, the encoders use the second strategy described above to correct the errors from the first step. The errors in the first step vanish as $\epsilon\to 0$, since the sources become equal with probability going to one. So, the rate of transmission approaches the rate of the first step which is close to $\log{q}-H(N_\delta)$. We provide a more detailed summary of the proof:
Fix $n$. Both encoders use an finite blocklength source-channel code for the $q$-ary symmetric channel with noise $N_{\delta}$ to transmit the sources. Let the blocklength of this code be equal to $n$, and the rate be equal to $\log{q}-H(N_\delta)+O(\frac{1}{\sqrt{n}})$. From the problem statement $P(X^n= Z^n)=P(N^n_{\epsilon}= 0^n)=(1-\epsilon)^n$. Since $U_{1}^n$ is a function of $X^n$, a and $U_2^n$ is a function of $Z^n$, we conclude that $P(U_1^n=U_2^n)\geq (1-\epsilon)^n$. The encoders then take turns to send refinements to the decoder. This is done using the second strategy described above. The rate required for this part of the transmission is  $\frac{(1-\epsilon)^n}{R_{s,1}}+\frac{H(N_{\epsilon})}{R_{s,1}}$. Note that $\frac{(1-\epsilon)^n}{R_{s,1}}+\frac{H(N_{\epsilon})}{R_{s,1}}$ goes to $0$ as $\epsilon\to 0$. 
This completes the proof. 
\end{proof}
 
 The following lemma provides an upper-bound to the entropy of $X$ as a function of $\epsilon$ and the correlation between $U_1$ and $U_2$.
 
 \begin{Lemma}
 For a coding scheme with encoding functions $U_1^n=\underline{e}_1(X^n), U_2^n=\underline{e}_{2}(Z^n)$, the following holds:
 \begin{align}
 \label{eq:Hx}
 H(X)\leq \frac{1}{n}\sum_{i=1}^n P(U_{1,i}=U_{2,i})\left(\log{q}-H(N_{\delta})\right)+1.
\end{align}
\end{Lemma}
\begin{proof}
Similar to the proof of Lemma  \ref{lem:corrIC}, we use Fano's inequality to prove the lemma.
Since $X^n$ is reconstructed losslessly at the decoder, by Fano's inequality the following holds:
 \begin{align*}
 &H(X^n)\approx I(U_1^nU_2^n;Y^n)\leq \sum_{i=1}^n  I(U_{1,i}U_{2,i};Y_{i})
 \stackrel{(a)}{=} \sum_{i=1}^n I(E_i,U_{1,i}U_{2,i};Y_n)=\sum_{i=1}^n I(E_i;Y_i)+I(U_{1,i}U_{2,i};Y_i|E_i)\\
 &\stackrel{(b)}\leq H(E_i)+\sum_{i=1}^n P(U_{1,i}=U_{2,i})(Y^n,Z^n))\left(\log{q}-H(N_\delta)\right)
 \\&\stackrel{(c)}{\leq} n+\sum_{i=1}^n P(U_{1,i}=U_{2,i})\left(\log{q}-H(N_\delta)\right),
\end{align*}
where in (a) we have defined $E_i$ as the indicator function of the event that $U_{1,i}=U_{2_i}, i\in[1,n]$, in (b) we have used that $I(U_{1,i}U_{2,i};Y_i|E_i)= P(E_i=0)\cdot 0+ P(E_i=1)I(U_{1,i};Y|U_{1,i}=U_{2,i})$ and in (c) we have used the fact that $E^n$ is binary.

\end{proof}
Since for SLCE's  $P(U_{1,i}=U_{2,i})$ is bounded away from 1 for $\epsilon\neq 0$, we conclude that there exists $q$ and $N_{\delta}$ such that $\frac{1}{n}\sum_{i=1}^n P(U_{1,i}=U_{2,i})\left(\log{q}-H(N_{\delta})\right)+1\leq \log{q}-H(N_{\delta}$. So, SLCE's are suboptimal in this example as well.

\section{Conclusion}\label{sec:conc}

We derived a new bound on the maximum correlation between Boolean functions operating on pairs of sequences of random variable. The bound was presented as a function of the dependency spectrum of the functions. We developed a new mathematical apparatus for analyzing Boolean functions, provided formulas for decomposing the Boolean function into additive components, and for calculating the dependency spectrum of these functions. The new bound may find applications in security, control and information theory.

Next, we characterized a set of properties which are shared between the SLCEs used in the literature. We showed that ensembles which have these properties produce encoding functions which are inefficient in preserving correlation. We derived a probabilistic upper-bound on the correlation between the outputs of random encoders generated using SLCEs. We showed that the correlation between the outputs of such encoders is discontinuous with respect to the input distribution. We used this discontinuity to show that all SLCSs are sub-optimal in two specific multi-terminal communications problem involving the transmission of correlated source.

\section{Appendix}
\subsection{Proof of Lemma \ref{prop:belong1}}
\begin{proof}
By definition, any element of $\mathcal{G}_{i_1}\otimes \mathcal{G}_{i_2}\otimes\dotsb \otimes \mathcal{G}_{i_n}$ satisfies the conditions in the proposition. Conversely, we show that any function satisfying the conditions $(1)$ and $(2)$ is in the tensor product.  Let $\tilde{f}=\sum_{\mathbf{j}}\tilde{f}_{\mathbf{j}}, \tilde{f}_{\mathbf{j}}\in\mathcal{G}_{j_1}\otimes \mathcal{G}_{j_2}\otimes\dotsb \otimes \mathcal{G}_{j_n}$ be an arbitrary function satisfying conditions $(1)$ and $(2)$ . Assume $i_k=1$ for some $k\in [1,n]$. Then: 
 \begin{align*}
 0\stackrel{(2)}{=}\mathbb{E}_{X^n|X_{\sim i_k}}(\sum_{\mathbf{j}}\tilde{f}_{\mathbf{j}}|X_{\sim i_k})\stackrel{(a)}{=}\sum_{\mathbf{j}}\mathbb{E}_{X^n|X_{\sim i_k}}(\tilde{f}_{\mathbf{j}}|X_{\sim i_k})\stackrel{(1)}{=} 
 \sum_{\mathbf{j}: j_k=0}\mathbb{E}_{X^n|X_{\sim i_k}}(\tilde{f}_{\mathbf{j}}|X_{\sim i_k})\stackrel{(2)}{=}\sum_{\mathbf{j}: j_k=0}\tilde{f}_{\mathbf{j}}, 
 \end{align*}
 where we have used linearity of expectation in (a), and the last two equalities use the fact that $\tilde{f}_{\mathbf{j}}\in \mathcal{G}_{j_1}\otimes \mathcal{G}_{j_2}\otimes\dotsb \otimes \mathcal{G}_{j_n}$ which means it satisfies properties $(1)$ and $(2)$. So far we have shown that $\tilde{f}=\sum_{\mathbf{j}\geq \mathbf{i}}\tilde{f}_{\mathbf{j}}$. Recall that $\mathbf{i}$ is given in the statement of the proposition. Now assume $i_{k'}=0$. Then:
 \begin{align*}
 \sum_{\mathbf{j}\geq\mathbf{i}}\tilde{f}_{\mathbf{j}}=\tilde{f}\stackrel{(1)}{=}\mathbb{E}_{X^n|X_{\sim i_{k'}}}(\sum_{\mathbf{j}\geq\mathbf{i}}\tilde{f}_{\mathbf{j}}|X_{\sim i_{k'}})=
 \sum_{\mathbf{j}\geq\mathbf{i}}\mathbb{E}_{X^n|X_{\sim i_{k'}}}(\tilde{f}_{\mathbf{j}}|X_{\sim i_{k'}})=  \sum_{\mathbf{j}\geq\mathbf{i}: j_{k'}=0}\tilde{f}_{\mathbf{j}} \Rightarrow 
 \sum_{\mathbf{j}\geq\mathbf{i}:j_{k'}=1}\tilde{f}_{\mathbf{j}}=0.
 \end{align*}
So, $\tilde{f}=\sum_{\mathbf{i}\geq\mathbf{j}\geq\mathbf{i}} \tilde{f}_{\mathbf{j}}=\tilde{f}_{\mathbf{i}}$. By assumption we have $\tilde{f}_{\mathbf{i}}\in\mathcal{G}_{i_1}\otimes \mathcal{G}_{i_2}\otimes\dotsb \otimes \mathcal{G}_{i_n}$. 
\end{proof}

\subsection{Proof of Lemma \ref{prop:belong2}}

\begin{proof}
 1) For two n-length binary vectors $\mathbf{i}$, and $\mathbf{j}$, we write $\mathbf{i}\leq \mathbf{j}$ if $i_k\leq j_k, \forall k\in[1,n]$. The set $\{0,1\}^n$ equipped with $\leq$ is a well-founded set (i.e. any subset of $\{0,1\}^n$ has at least one minimal element).  The following presents the principle of Noetherian induction on well-founded sets:
 \begin{Proposition}[Principle of Noetherian Induction \cite{noetherian}]
 Let $(A,\preccurlyeq)$ be a well-founded set. To prove the property $P(x)$ is true for all elements $x$ in $A$, it is sufficient to prove the following
 \\1) \textbf{Induction Basis:} $P(x)$ is true for all minimal elements in $A$. 
 \\2) \textbf{Induction Step:} For any non-minimal element $x$ in $A$, if $P(y)$ is true for all minimal $y$ such that $y\prec x$, then it is true for $x$.
 \end{Proposition}
 We will use Noetherian induction to prove the result.
  Let $\mathbf{i}_j, j\in [1,n]$ be the $j$th element of the standard basis. Then $\tilde{e}_{\mathbf{i}_j}= \mathbb{E}_{X^n|X_{ j}}(\tilde{e}|X_{ j})$. By the smoothing property of expectation, $\mathbb{E}_{X^n}(\tilde{e}_{\mathbf{i}_j})=\mathbb{E}_{X^n}(\tilde{e})=0$. Assume that $\forall \mathbf{j}<\mathbf{i}$,  $\mathbb{E}_{X^n}(\tilde{e}_{\mathbf{j}})=0$. Then,
\begin{align*}
 \mathbb{E}_{X^n}(\tilde{e}_{\mathbf{i}})
&=\mathbb{E}_{X^n}\left(\mathbb{E}_{X^n|X_{\mathbf{i}}}(\tilde{e}|X_{\mathbf{i}})-\sum_{\mathbf{j}< \mathbf{i}} \tilde{e}_{\mathbf{j}}\right)
\\&= 
\mathbb{E}_{X^n}(\tilde{e})- \sum_{\mathbf{j}< \mathbf{i}}  \mathbb{E}_{X^n}(\tilde{e}_{\mathbf{j}})
=0- \sum_{\mathbf{j}< \mathbf{i}}  0 =0.
\end{align*}
\\2) This statement is also proved by induction. $\mathbb{E}_{X^n|X_{\mathbf{i}}}(\tilde{e}|X_{\mathbf{i}})$ is a function of $X_{\mathbf{i}}$, so by induction $\tilde{e}_{\mathbf{i}}=\mathbb{E}_{X^n|X_{\mathbf{i}}}(\tilde{e}|X_{\mathbf{i}})-\sum_{\mathbf{j}<\mathbf{i}}\tilde{e}_{\mathbf{k}}$ is also a function of $X_{\mathbf{i}}$. 
\\3) Let $\mathbf{i}_k, k\in [1,n]$ be defined as the $k$th element of the standard basis, and take $j,j'\in [1,n], j\neq j'$. We have:
\begin{align*}
 \mathbb{E}_{X^n}(\tilde{e}_{\mathbf{i}_j}\tilde{e}_{\mathbf{i}_{j'}})=\mathbb{E}_{X^n}(\mathbb{E}_{X^n|X_j}(\tilde{e}|X_j)\mathbb{E}_{X^n|X_{j'}}(\tilde{e}|X_{j'}))
 \stackrel{(a)}{=} \mathbb{E}_{X^n}(\mathbb{E}_{X^n|X_j}(\tilde{e}|X_j)) \mathbb{E}_{X^n}(\mathbb{E}_{X^n|X_{j'}}(\tilde{e}|X_{j'}))
 \stackrel{(b)}{=} \mathbb{E}^2_{X^n}(\tilde{e})=0,
\end{align*}
where we have used the memoryless property of the source in (a) and (b) results from the smoothing property of expectation. We extend the argument by Noetherian induction. Fix $\mathbf{i}, \mathbf{k}$. Assume that $\mathbb{E}_{X^n}(\tilde{e}_{\mathbf{j}}\tilde{e}_{\mathbf{j}'})=\mathbbm{1}(\mathbf{j}=\mathbf{j}')\mathbb{E}_{X^n}(\tilde{e}^2_{\mathbf{j}}), \forall  \mathbf{j}< \mathbf{i}, \mathbf{j}'\leq \mathbf{k}$, and $\forall  \mathbf{j}\leq \mathbf{i}, \mathbf{j}'\leq \mathbf{k}$. Then, we have 
\begin{align*}
  &\mathbb{E}_{X^n}(\tilde{e}_{\mathbf{i}}\tilde{e}_{\mathbf{k}})
  =\mathbb{E}_{X^n}\left(\left(\mathbb{E}_{X^n|X_{\mathbf{i}}}(\tilde{e}|X_{\mathbf{i}})-\sum_{\mathbf{j}< \mathbf{i}} \tilde{e}_{\mathbf{j}}\right)\left(\mathbb{E}_{X^n|X_{\mathbf{k}}}(\tilde{e}|X_{\mathbf{k}})-\sum_{\mathbf{j}'< \mathbf{k}} \tilde{e}_{\mathbf{j}'}\right)\right)\\
  &=\mathbb{E}_{X_n}\left(\mathbb{E}_{X^n|X_{\mathbf{i}}}(\tilde{e}|X_{\mathbf{i}})\mathbb{E}_{X^n|X_{\mathbf{k}}}(\tilde{e}|X_{\mathbf{k}})\right)
  -\sum_{\mathbf{j}< \mathbf{i}} \mathbb{E}_{X^n}\left(\tilde{e}_{\mathbf{j}}\mathbb{E}_{X^n|X_{\mathbf{k}}}(\tilde{e}|X_{\mathbf{k}})\right)
  - \sum_{\mathbf{j}'< \mathbf{k}} \mathbb{E}_{X^n}\left(\tilde{e}_{\mathbf{j}'}\mathbb{E}_{X^n|X_{\mathbf{i}}}(\tilde{e}|X_{\mathbf{i}})\right)
  + \sum_{\mathbf{j}< \mathbf{i}, \mathbf{j}'<\mathbf{k}}\mathbb{E}_{X^n}(\tilde{e}_{\mathbf{j}}\tilde{e}_{\mathbf{j}'}).
\end{align*}
The second and third terms in the above expression can be simplified as follows. First, note that:
\begin{align}
&\tilde{e}_{\mathbf{i}}=\mathbb{E}_{X^n|X_{\mathbf{i}}}(\tilde{e}|X_{\mathbf{i}})-\sum_{\mathbf{j}< \mathbf{i}} \tilde{e}_{\mathbf{j}}\Rightarrow \sum_{\mathbf{j}\leq \mathbf{i}} \tilde{e}_{\mathbf{j}}=\mathbb{E}_{X^n|X_{\mathbf{i}}}(\tilde{e}|X_{\mathbf{i}}).
\label{eq:sum_int}
\end{align}
Our goal is to simplify $\mathbb{E}_{X^n}(\tilde{e}_{\mathbf{j}}\mathbb{E}_{X^n|X_{\mathbf{j}'}}(\tilde{e}|X_{\mathbf{j}'}))$. We proceed by considering two different cases:
\\\textbf{Case 1:} $\mathbf{i}\nleq \mathbf{k}$ and $\mathbf{k}\nleq \mathbf{i}$: 

Let $\mathbf{j}<\mathbf{i}$:
\begin{align*}
\mathbb{E}_{X^n}(\tilde{e}_{\mathbf{j}}\mathbb{E}_{X^n|X_{\mathbf{k}}}(\tilde{e}|X_{\mathbf{k}}))
\stackrel{\eqref{eq:sum_int}}{=} \mathbb{E}_{X^n}(\tilde{e}_{\mathbf{j}}\sum_{\mathbf{l}\leq \mathbf{k}} \tilde{e}_{\mathbf{l}}))=\sum_{\mathbf{l}\leq \mathbf{k}} \mathbb{E}_{X^n}(\tilde{e}_{\mathbf{j}}\tilde{e}_{\mathbf{l}})=\sum_{\mathbf{l}\leq \mathbf{k}}\mathbbm{1}({\mathbf{j}=\mathbf{l}})  \mathbb{E}_{X^n}(\tilde{e}_{\mathbf{j}}^2)=\mathbbm{1}({\mathbf{j}\leq\mathbf{k}})  \mathbb{E}_{X^n}(\tilde{e}_{\mathbf{j}}^2).
\end{align*}
By the same arguments, for $\mathbf{j}'\leq\mathbf{k}$:
\begin{align*}
  \mathbb{E}_{X^n}\left(\tilde{e}_{\mathbf{j}'}\mathbb{E}_{X^n|X_{\mathbf{i}}}(\tilde{e}|X_{\mathbf{i}})\right)=
  \mathbbm{1}({\mathbf{j}'\leq\mathbf{i}})  \mathbb{E}_{X^n}(\tilde{e}_{\mathbf{j}'}^2).
\end{align*}

Replacing the terms in the original equality we get:
\begin{align}
   \mathbb{E}_{X^n}(\tilde{e}_{\mathbf{i}}\tilde{e}_{\mathbf{k}})
   &= \mathbb{E}_{X^n}\left(\mathbb{E}_{X^n|X_{\mathbf{i}}}(\tilde{e}|X_{\mathbf{i}})\mathbb{E}_{X^n|X_{\mathbf{k}}}(\tilde{e}|X_{\mathbf{k}})\right)
   -\sum_{\mathbf{j}< \mathbf{i}}\mathbbm{1}({\mathbf{j}\leq\mathbf{k}})  \mathbb{E}_{X^n}(\tilde{e}_{\mathbf{j}}^2)
   - \sum_{\mathbf{j}'\leq \mathbf{k}} \mathbbm{1}({\mathbf{j}'\leq\mathbf{i}})  \mathbb{E}_{X^n}(\tilde{e}_{\mathbf{j}'}+ \sum_{\mathbf{j}< \mathbf{i}, \mathbf{j}'<\mathbf{k}}\mathbbm{1}(\mathbf{j}=\mathbf{j}')\mathbb{E}_{X^n}(\tilde{e}_{\mathbf{j}}^2).
   \label{eq:fix}
   \end{align}
   
   Note that: 
   \begin{align*}
& \sum_{\mathbf{j}< \mathbf{i}}\mathbbm{1}({\mathbf{j}\leq\mathbf{k}})  \mathbb{E}_{X^n}(\tilde{e}_{\mathbf{j}}^2)=
 \sum_{\mathbf{j}< \mathbf{i}, \mathbf{j}\leq\mathbf{k}} \mathbb{E}_{X^n}(\tilde{e}_{\mathbf{j}}^2)=
 \sum_{\mathbf{j}\leq \mathbf{i}\wedge \mathbf{k}}\mathbb{E}_{X^n}(\tilde{e}_{\mathbf{j}}^2)
 \\&
 \sum_{\mathbf{j}'\leq \mathbf{k}} \mathbbm{1}({\mathbf{j}'\leq\mathbf{i}})  \mathbb{E}_{X^n}(\tilde{e}_{\mathbf{j}'}=
  \sum_{\mathbf{j}'\leq \mathbf{k},{\mathbf{j}'\leq\mathbf{i}}} \mathbb{E}_{X^n}(\tilde{e}_{\mathbf{j}'}=
 \sum_{\mathbf{j}\leq \mathbf{i}\wedge \mathbf{k}}\mathbb{E}_{X^n}(\tilde{e}_{\mathbf{j}}^2)  
 \\&
 \sum_{\mathbf{j}< \mathbf{i}, \mathbf{j}'<\mathbf{k}}\mathbbm{1}(\mathbf{j}=\mathbf{j}')\mathbb{E}_{X^n}(\tilde{e}_{\mathbf{j}}^2=
 \sum_{\mathbf{j}\leq \mathbf{i}\wedge \mathbf{k}}\mathbb{E}_{X^n}(\tilde{e}_{\mathbf{j}}^2) 
\end{align*}
   
   Replacing the terms in \eqref{eq:fix}, we have: 
   \begin{align*}
 &  \mathbb{E}_{X^n}(\tilde{e}_{\mathbf{i}}\tilde{e}_{\mathbf{k}})= \mathbb{E}_{X^n}\left(\mathbb{E}_{X^n|X_{\mathbf{i}}}(\tilde{e}|X_{\mathbf{i}})\mathbb{E}_{X^n|X_{\mathbf{k}}}(\tilde{e}|X_{\mathbf{k}})\right)-\sum_{\mathbf{j}\leq \mathbf{i}\wedge \mathbf{k}}\mathbb{E}_{X^n}(\tilde{e}_{\mathbf{j}}^2)\\
   &\stackrel{(a)}{=}\mathbb{E}_{X^n}(\mathbb{E}_{X^n|X_{ \mathbf{i}\wedge \mathbf{k}}}^2(\tilde{e}(X^n)|X_{ \mathbf{i}\wedge \mathbf{k}}))-\sum_{\mathbf{j}\leq \mathbf{i}\wedge \mathbf{k}}\mathbb{E}_{X^n}(\tilde{e}_{\mathbf{j}}^2)
   \\&\stackrel{(b)}{=}\mathbb{E}_{X^n}(\mathbb{E}_{X^n|X_{ \mathbf{i}\wedge \mathbf{k}}}^2(\tilde{e}(X^n)|X_{ \mathbf{i}\wedge \mathbf{k}}))-\mathbb{E}_{X^n}\left((\sum_{\mathbf{j}\leq \mathbf{i}\wedge \mathbf{k}}\tilde{e}_{\mathbf{j}})^2\right)\stackrel{\eqref{eq:sum_int}}{=}0,
\end{align*}
where in (b) we have used that $\tilde{e}_{\mathbf{i}}$'s are uncorrelated, and (a) is proved below:
\begin{align*}
\mathbb{E}_{X^n}\left(\mathbb{E}_{X^n|X_{\mathbf{i}}}(\tilde{e}|X_{\mathbf{i}})\mathbb{E}_{X^n|X_{\mathbf{k}}}(\tilde{e}|X_{\mathbf{k}})\right)
&= \sum_{x_{ \mathbf{i}\wedge \mathbf{k}}}P(x_{\mathbf{i}\wedge \mathbf{k}})
\left( \left(\sum_{x_{{|\mathbf{i}}- \mathbf{k}|^+}}P(x_{|{\mathbf{i}}- \mathbf{k}|^+})
\mathbb{E}_{X^n|X_{\mathbf{i}}}(\tilde{e}|X_{\mathbf{i}})\right)\left(\sum_{x_{{|\mathbf{k}}- \mathbf{i}|^+}}P(x_{|{\mathbf{k}}- \mathbf{i}|^+})\mathbb{E}_{X^n|X_{\mathbf{k}}}(\tilde{e}|X_{\mathbf{k}})\right)\right)\\
&=\sum_{x_{ \mathbf{i}\wedge \mathbf{k}}}P(x_{\mathbf{i}\wedge \mathbf{k}}) \mathbb{E}_{X^n|X_{\mathbf{i}\wedge\mathbf{k}}}^2(\tilde{e}|x_{\mathbf{i}\wedge\mathbf{k}})\\
&=\mathbb{E}_{X^n}(\mathbb{E}_{X^n|X_{ \mathbf{i}\wedge \mathbf{k}}}^2(\tilde{e}(X^n)|X_{ \mathbf{i}\wedge \mathbf{k}})).
\end{align*}
\\\textbf{Case 2:} Assume $\mathbf{i}\leq \mathbf{k}$:
\begin{align*}
    \mathbb{E}_{X^n}(\tilde{e}_{\mathbf{i}}\tilde{e}_{\mathbf{k}})&=  \mathbb{E}_{X^n}\left(\mathbb{E}_{X^n|X_{\mathbf{i}}}(\tilde{e}|X_{\mathbf{i}})\mathbb{E}_{X^n|X_{\mathbf{k}}}(\tilde{e}|X_{\mathbf{k}})\right)-\sum_{\mathbf{j}< \mathbf{i}}\mathbbm{1}({\mathbf{j}\leq\mathbf{k}})  \mathbb{E}_{X^n}(\tilde{e}_{\mathbf{j}}^2)
    \\&- \sum_{\mathbf{j}'< \mathbf{k}} \mathbbm{1}({\mathbf{j}'\leq\mathbf{i}}) 
     \mathbb{E}_{X^n}(\tilde{e}_{\mathbf{j}'}^2) + \sum_{\mathbf{j}< \mathbf{i}, \mathbf{j}'<\mathbf{k}}\mathbbm{1}(\mathbf{j}=\mathbf{j}')
     \mathbb{E}_{X^n}(\tilde{e}_{\mathbf{j}}^2)\\
    &\stackrel{(a)}{=}\mathbb{E}_{X^n}(\mathbb{E}_{X^n|X_{\mathbf{i}}}^2(\tilde{e}|X_{\mathbf{i}}))-\sum_{\mathbf{j}<\mathbf{i}}\mathbb{E}_{X^n}(\tilde{e}^2_{\mathbf{j}})-\sum_{\mathbf{j}'\leq\mathbf{i}}\mathbb{E}_{X^n}(\tilde{e}^2_{\mathbf{j}'})+\sum_{\mathbf{j}\leq\mathbf{i}}\mathbb{E}_{X^n}(\tilde{e}^2_{\mathbf{j}})\\
    &=0,
\end{align*}
where in (a) we have used (a) proved above.
\\\textbf{Case 3:} When $\mathbf{k}\leq \mathbf{i}$ the proof is similar to case 2.
\\4) Clearly when $|\mathbf{i}|=1$, the claim holds. Assume it is true for all $\mathbf{j}$ such that $|\mathbf{j}|<|\mathbf{i}|$. 
Take $\mathbf{i}\in \{0,1\}^n$ and $t\in [1,n], i_t=1$ arbitrarily. We first prove the claim for $\mathbf{k}=\mathbf{i}-\mathbf{i}_t$:
\begin{align*}
 &\mathbb{E}_{X^n|X_{\mathbf{k}}}(\tilde{e}_{\mathbf{i}}|X_{\mathbf{k}})= \mathbb{E}_{X^n|X_{\mathbf{k}}}\left(\left(\mathbb{E}_{X^n|X_{\mathbf{i}}}(\tilde{e}|X_{\mathbf{i}})-\sum_{\mathbf{j}<\mathbf{i}}\tilde{e}_{\mathbf{j}}\right)|X_{\mathbf{k}}\right)
 = \mathbb{E}_{X^n|X_{\mathbf{k}}}\left(\mathbb{E}_{X^n|X_{\mathbf{i}}}(\tilde{e}|X_{\mathbf{i}})|X_{\mathbf{k}}\right)-\sum_{\mathbf{j}<\mathbf{i}}\mathbb{E}_{X^n|X_{\mathbf{k}}}(\tilde{e}_{\mathbf{j}}|X_{\mathbf{k}})
 \\&\stackrel{(a)}{=}\mathbb{E}_{X^n|X_{\mathbf{k}}}(\tilde{e}|X_{\mathbf{k}})-\sum_{\mathbf{j}<\mathbf{i}}\mathbb{E}_{X^n|X_{\mathbf{k}}}(\tilde{e}_{\mathbf{j}}|X_{\mathbf{k}})
 \stackrel{(b)}{=} \sum_{\mathbf{j}\leq \mathbf{i}-\mathbf{i}_t}\tilde{e}_{\mathbf{j}} -\sum_{\mathbf{j}<\mathbf{i}}\mathbb{E}_{X^n|X_{\mathbf{k}}}(\tilde{e}_{\mathbf{j}}|X_{\mathbf{k}})
 \\&\stackrel{(c)}{=} \sum_{\mathbf{j}\leq \mathbf{i}-\mathbf{i}_t}\mathbb{E}_{X^n|X_{\mathbf{k}}}(\tilde{e}_{\mathbf{j}}|X_{\mathbf{k}}) -\sum_{\mathbf{j}<\mathbf{i}}\mathbb{E}_{X^n|X_{\mathbf{k}}}(\tilde{e}_{\mathbf{j}}|X_{\mathbf{k}})
 =\sum_{s\neq t}\mathbb{E}_{X^n|X_{\mathbf{k}}}(\tilde{e}_{\mathbf{i}-\mathbf{i}_s}|X_{\mathbf{k}})
 \stackrel{(d)} {=}\sum_{s\neq t}\mathbb{E}_{X^n|X_{\mathbf{k}-\mathbf{i}_s}}(\tilde{e}_{\mathbf{i}-\mathbf{i}_s}|X_{\mathbf{k}-\mathbf{i}_s})
 \stackrel{(e)}=0,
\end{align*}
where in (a) we have used $\mathbf{i}>\mathbf{k}$, (b) follows from equation {\eqref{eq:sum_int}}, also (c) follows from $\mathbf{j}<\mathbf{k}$, (e) uses $\mathbf{k}\wedge(\mathbf{i}-\mathbf{i}_s)= \mathbf{k}-\mathbf{i}_s$, and finally, (d) uses the induction assumption. Now we extend the result to general $\mathbf{k}<\mathbf{i}$. Fix  $\mathbf{k}$. Assume the claim is true for all $\mathbf{j}$ such that $\mathbf{k}<\mathbf{j}<\mathbf{i}$ (i.e $\forall \mathbf{k}<\mathbf{j}<\mathbf{i}, \mathbb{E}_{X^n|X_{\mathbf{k}}}(\tilde{e}_{X_{\mathbf{j}}|X_{\mathbf{k}}})=0$). We have:
\begin{align*}
 &\mathbb{E}_{X^n|X_{\mathbf{k}}}(\tilde{e}_{\mathbf{i}}|X_{\mathbf{k}})= \mathbb{E}_{X^n|X_{\mathbf{k}}}\left(\mathbb{E}_{X^n|X_{\mathbf{i}}}(\tilde{e}|X_{\mathbf{i}})-\sum_{\mathbf{j}<\mathbf{i}}\tilde{e}_{\mathbf{j}}|X_{\mathbf{k}}\right)
 = \mathbb{E}_{X^n|X_{\mathbf{k}}}\left(\mathbb{E}_{X^n|X_{\mathbf{i}}}(\tilde{e}|X_{\mathbf{i}})|X_{\mathbf{k}}\right)-\sum_{\mathbf{j}\leq \mathbf{k}}\mathbb{E}_{X^n|X_{\mathbf{k}}}(\tilde{e}_{\mathbf{j}}|X_{\mathbf{k}})
\\& =\mathbb{E}_{X^n|X_{\mathbf{k}}}(\tilde{e}|X_{\mathbf{k}})-\sum_{\mathbf{j}\leq \mathbf{k}}\tilde{e}_{\mathbf{j}}\stackrel{\eqref{eq:sum_int}}{=}0.
\end{align*}
\end{proof}

\subsection{Proof of Proposition \ref{Lem:power}}
\begin{proof}
 \begin{align*}
 \mathbf{P}_{\mathbf{i}}
 &=Var_{X_{\mathbf{i}}}(\tilde{e}_{\mathbf{i}}(X^n))
 = \mathbb{E}_{X_{\mathbf{i}}}(\tilde{e}^2_{\mathbf{i}}(X^n))-\mathbb{E}_{X_{\mathbf{i}}}^2(\tilde{e}_{\mathbf{i}}(X^n))\\
& \stackrel{(a)}{=} \mathbb{E}_{X_{\mathbf{i}}}\left(\left(\mathbb{E}_{X^n|X_{\mathbf{i}}}(\tilde{e}|X_{\mathbf{i}})-\sum_{\mathbf{j}< \mathbf{i}} \tilde{e}_{\mathbf{j}}\right)^2\right)-0
=\mathbb{E}_{X_{\mathbf{i}}}\left(\mathbb{E}_{X^n|X_{\mathbf{i}}}^2(\tilde{e}|X_{\mathbf{i}})\right)-2\sum_{\mathbf{j}<\mathbf{i}}\mathbb{E}_{X_{\mathbf{i}}}\left(\mathbb{E}_{X^n|X_{\mathbf{i}}}(\tilde{e}|X_{\mathbf{i}})\tilde{e}_{\mathbf{j}}\right)+\mathbb{E}_{X_{\mathbf{i}}}\left(\left(\sum_{\mathbf{j}< \mathbf{i}} \tilde{e}_{\mathbf{j}}\right)^2\right)
\\&\stackrel{(b)}=\mathbb{E}_{X_{\mathbf{i}}}\left(\mathbb{E}_{X^n|X_{\mathbf{i}}}^2(\tilde{e}|X_{\mathbf{i}})\right)-2\sum_{\mathbf{j}<\mathbf{i}}\mathbb{E}_{X_{\mathbf{i}}}\left(\mathbb{E}_{X^n|X_{\mathbf{i}}}(\sum_{\mathbf{l}\leq \mathbf{i}}\tilde{e}_{\mathbf{l}}|X_{\mathbf{i}})\tilde{e}_{\mathbf{j}}\right)+\mathbb{E}_{X_{\mathbf{i}}}\left(\left(\sum_{\mathbf{j}< \mathbf{i}} \tilde{e}_{\mathbf{j}}\right)^2\right)
\\&\stackrel{(c)}{=}\mathbb{E}_{X_{\mathbf{i}}}\left(\mathbb{E}_{X^n|X_{\mathbf{i}}}^2(\tilde{e}|X_{\mathbf{i}})\right)-2\sum_{\mathbf{j}<\mathbf{i}}\mathbb{E}_{X_{\mathbf{i}}}\left(\sum_{\mathbf{l}\leq \mathbf{i}}\mathbb{E}_{X^n|X_{\mathbf{i}}}(\tilde{e}_{\mathbf{l}}|X_{\mathbf{i}})\tilde{e}_{\mathbf{j}}\right)+\mathbb{E}_{X_{\mathbf{i}}}\left(\left(\sum_{\mathbf{j}< \mathbf{i}} \tilde{e}_{\mathbf{j}}\right)^2\right)
\\&\stackrel{(d)}{=}\mathbb{E}_{X_{\mathbf{i}}}\left(\mathbb{E}_{X^n|X_{\mathbf{i}}}^2(\tilde{e}|X_{\mathbf{i}})\right)-2\sum_{\mathbf{j}<\mathbf{i}}\mathbb{E}_{X_{\mathbf{i}}}\left(\sum_{\mathbf{l}<\mathbf{i}}\tilde{e}_{\mathbf{l}}\tilde{e}_{\mathbf{j}}\right)+\mathbb{E}_{X_{\mathbf{i}}}\left(\left(\sum_{\mathbf{j}< \mathbf{i}} \tilde{e}_{\mathbf{j}}\right)^2\right)
\\&\stackrel{(e)}{=}\mathbb{E}_{X_{\mathbf{i}}}\left(\mathbb{E}_{X^n|X_{\mathbf{i}}}^2(\tilde{e}|X_{\mathbf{i}})\right)-2\sum_{\mathbf{j}<\mathbf{i}}\sum_{\mathbf{l}<\mathbf{i}}\mathbbm{1}(\mathbf{j}=\mathbf{l})\mathbb{E}_{X_{\mathbf{i}}}\left(\tilde{e}_{\mathbf{l}}\tilde{e}_{\mathbf{j}}\right)+\mathbb{E}_{X_{\mathbf{i}}}\left(\left(\sum_{\mathbf{j}< \mathbf{i}} \tilde{e}_{\mathbf{j}}\right)^2\right)
\\&=\mathbb{E}_{X_{\mathbf{i}}}\left(\mathbb{E}_{X^n|X_{\mathbf{i}}}^2(\tilde{e}|X_{\mathbf{i}})\right)-2\sum_{\mathbf{j}<\mathbf{i}}\mathbb{E}_{X_{\mathbf{j}}}(\tilde{e}^2_{\mathbf{j}})+\mathbb{E}_{X_{\mathbf{i}}}\left(\left(\sum_{\mathbf{j}< \mathbf{i}} \tilde{e}_{\mathbf{j}}\right)^2\right)
 \\&=\mathbb{E}_{X_{\mathbf{i}}}(\mathbb{E}_{X^n|X_{\mathbf{i}}}^2(\tilde{e}|X_{\mathbf{i}}))-2\sum_{\mathbf{j}<\mathbf{i}}\mathbb{E}_{X_{\mathbf{j}}}(\tilde{e}^2_{\mathbf{j}})+\sum_{\mathbf{j}< \mathbf{i}} \sum_{\mathbf{k}<\mathbf{i}}\mathbb{E}_{X_{\mathbf{i}}}(\tilde{e}_{\mathbf{j}}\tilde{e}_{\mathbf{k}})
 \\&
\stackrel{(f)}{=}\mathbb{E}_{X_{\mathbf{i}}}(\mathbb{E}_{X^n|X_{\mathbf{i}}}^2(\tilde{e}|X_{\mathbf{i}}))-2\sum_{\mathbf{j}<\mathbf{i}}\mathbb{E}_{X_{\mathbf{j}}}(\tilde{e}^2_{\mathbf{j}})+\sum_{\mathbf{j}< \mathbf{i}} \sum_{\mathbf{k}<\mathbf{i}}\mathbbm{1}(\mathbf{j}=\mathbf{k})\mathbb{E}_{X_{\mathbf{i}}}(\tilde{e}^2_{\mathbf{j}})
 \\&=\mathbb{E}_{X_{\mathbf{i}}}(\mathbb{E}_{X^n|X_{\mathbf{i}}}^2(\tilde{e}|X_{\mathbf{i}}))-\sum_{\mathbf{j}< \mathbf{i}}\mathbf{P}_{\mathbf{j}},
\end{align*}
\end{proof}
where (a) follows from condition 1) in Lemma \ref{prop:belong2}, b) follows from the decomposition in Equation \eqref{eq:sum_int} in the appendix, (c) uses linearity of expectation, (d) holds from condition 2) in Lemma \ref{prop:belong2}, and in (e) and (f) we have used condition 1) in Lemma \ref{prop:belong2}.

\subsection{Proof of Theorem \ref{th:sec3}}
\begin{proof}
This proof builds upon the results in \cite{ComInf2}. The proof involves three main steps. In the first two steps we prove the lower bound. First, we bound the Pearson correlation \cite{Pearson} between the real-valued functions $\tilde{e}$, and $\tilde{f}$. In the second step, we relate the correlation to the probability that the two functions are equal and derive the necessary bounds. Finally, in the third step we use the lower bound proved in the first two steps to derive the upper bound. 
\\\textbf{Step 1:} Let  $ s\triangleq P_X(e(X^n)=1)$, $r\triangleq P_Y(f(Y^n)=1)$. From Remark \ref{Rem:exp_0}, the expectation of both functions is 0. So, the Pearson correlation is given by 
\[\frac{\mathbb{E}_{X^n,Y^n}(\tilde{e}\tilde{f})}{\left(rs(1-s)(1-r)\right)^{\frac{1}{2}}}.\] Our goal is to bound this value. We have:
\begin{align}
 &\mathbb{E}_{X^n,Y^n}(\tilde{e}\tilde{f})
 \stackrel{(a)}=\mathbb{E}_{X^n,Y^n}\left((\sum_{\mathbf{i}\in\{0,1\}^n}\tilde{e}_{\mathbf{i}})(\sum_{\mathbf{k}\in\{0,1\}^n}\tilde{f}_{\mathbf{k}})\right)
 \stackrel{(b)}{=}\sum_{\mathbf{i}\in\{0,1\}^n}\sum_{\mathbf{k}\in\{0,1\}^n}\mathbb{E}_{X^n,Y^n}( \tilde{e}_{\mathbf{i}}\tilde{f}_{\mathbf{k}}).
 \label{Eq:init1}
\end{align}
In (a) we have used Definition \ref{Rem:Dec}, and in (b) we use linearity of expectation.
 Using the fact that $\tilde{e}_{\mathbf{i}}\in \mathcal{G}_{i_1}\otimes \mathcal{G}_{i_2}\otimes\dotsb \otimes \mathcal{G}_{i_n}$ and Definition \ref{Lem:tensor_dec}, we have:
 
\begin{align}
 \tilde{e}_{\mathbf{i}}=c_{\mathbf{i}}\prod_{t:i_t=1} \tilde{h}(X_{t}),  \tilde{f}_{\mathbf{k}}=d_{\mathbf{k}}\prod_{t:k_t=1} \tilde{g}(Y_{t}).
 \label{eq:init2}
\end{align}
where, 
\begin{align}
\tilde{h}(X)=    \begin{cases}
      1-q, & \text{if } X=1, \\
      -q. & \text{if } X=0,
    \end{cases},\qquad 
\tilde{g}(Y)=    \begin{cases}
      1-r, & \text{if } Y=1, \\
      -r. & \text{if } Y=0,
    \end{cases}
    \label{eq:basis} 
\end{align}

We replace $\tilde{e}_{\mathbf{i}}$ and $ \tilde{f}_{\mathbf{k}}$ in \eqref{Eq:init1}: 
 
\begin{align}
& \mathbb{E}_{X^n,Y^n}( \tilde{e}_{\mathbf{i}}\tilde{f}_{\mathbf{k}})
\stackrel{\eqref{eq:init2}}{=}\mathbb{E}_{X^n,Y^n}\left(\left(c_{\mathbf{i}} \prod_{t:i_t=1} \tilde{h}(X_{t})\right)\left(d_{\mathbf{k}} \prod_{s:k_s=1} \tilde{g}(Y_{s})\right)\right)
\stackrel{(a)}{=}c_{\mathbf{i}}d_{\mathbf{k}}  \mathbb{E}_{X^n,Y^n}\left(\prod_{t:i_t=1}\tilde{h}(X_{t}) \prod_{s:k_s=1} \tilde{g}(Y_{s})\right)\nonumber
\\&\stackrel{(b)}{=}c_{\mathbf{i}}d_{\mathbf{k}} \mathbb{E}_{X^n,Y^n}\left(\prod_{t:i_t=1, k_t=1}\tilde{h}(X_{t})\tilde{g}(Y_{k_t})\right)
\mathbb{E}_{X^n}\left(\prod_{t:i_t=1, k_t=0} \tilde{h}(X_{t})\right)\mathbb{E}_{Y^n}\left( \prod_{t:i_t=0, k_t=1} \tilde{g}(Y_{k_t})\right)\nonumber\\
&\stackrel{(c)}{=}\mathbbm{1}(\mathbf{i}=\mathbf{k})c_{\mathbf{i}}d_{\mathbf{k}}\prod_{t:i_t=1} \mathbb{E}_{X^n,Y^n}\left(\tilde{h}(X_{t}) \tilde{g}(Y_{t})\right)
\stackrel{(d)}{\leq} \mathbbm{1}(\mathbf{i}=\mathbf{k})c_{\mathbf{i}}d_{\mathbf{k}}  (1-2\epsilon)^{N_{\mathbf{i}}} \prod_{t:i_t=1}\mathbb{E}_{X^n}^{\frac{1}{2}}\left(\tilde{e}^2(X_{t})\right)\mathbb{E}_{Y^n}^{\frac{1}{2}}\left( \tilde{g}^2(Y)\right)\nonumber\\
&\stackrel{(e)}{=} \mathbbm{1}(\mathbf{i}=\mathbf{k}) (1-2\epsilon)^{N_{\mathbf{i}}}\mathbf{P}^{\frac{1}{2}}_{\mathbf{i}}\mathbf{Q}^{\frac{1}{2}}_{\mathbf{i}}
=  \mathbbm{1}(\mathbf{i}=\mathbf{k}) C_{\mathbf{i}}\mathbf{P}^{\frac{1}{2}}_{\mathbf{i}}\mathbf{Q}^{\frac{1}{2}}_{\mathbf{i}}.
\label{eq:intemid1}
\end{align}
(a) follows from linearity of expectation. In (b) we have used the fact that in a pair of DMS's, $X_i$ and $Y_j$ are independent for $i\neq j$. (c) holds since from Lemma \ref{prop:belong2}, $\mathbb{E}(\tilde{e}_{i})=\mathbb{E}(\tilde{f}_i)=0, \forall i\in [1,n]$. We prove (d) in Lemma \ref{Lem:init} below. In (e) we have used proposition \ref{pr:partfun}.
 
\begin{Lemma}
\label{Lem:init}
 Let $g(X)$ and $h(Y)$ be two arbitrary zero-mean, real valued functions, then:
 
\begin{align*}
 \mathbb{E}_{X,Y}(g(X)h(Y))\leq (1-2\epsilon)\mathbb{E}_X^{\frac{1}{2}}(g^2(X))\mathbb{E}_Y^{\frac{1}{2}}(h^2(Y)).
\end{align*}
 \end{Lemma}
\begin{proof}
This is a well-known result \cite{Ananth}. A proof is provided here for completeness:
 Let the functions be given as follows:
\begin{align*}
g(X)=    \begin{cases}
      \alpha & \quad \text{if } X=0 \\
      \beta & \quad \text{if } X=1.
    \end{cases} 
,\qquad h(Y)=    \begin{cases}
      \gamma &\quad \text{if }  Y=0 \\
      \delta & \quad \text{if } Y=1.
    \end{cases} 
\end{align*}
Also, let $P(X=1)=p$, and $P(Y=1)=q$. The zero-mean condition enforces the following equalities:
\begin{align*}
 &\alpha(1-p)+\beta p=0\Rightarrow \beta=\frac{-(1-p)\alpha}{p},\qquad\qquad\gamma(1-q)+\delta q=0 \Rightarrow \delta=\frac{-(1-q)\gamma}{q}.
\end{align*}
Next, we calculate the joint distribution of $P_{XY}$. Let $P_{i,j}\triangleq P(X=i,Y=j), i,j\in \{0,1\}$. We have the following:
\begin{align*}
 &P_{0,0}+P_{0,1}=P(X=0)=1-p,\qquad P_{0,0}+P_{1,0}=P(Y=0)=1-q,\\
 &P_{0,0}+P_{1,1}=P(X=Y)=1-\epsilon,\qquad P_{0,0}+P_{0,1}+P_{1,0}+P_{1,1}=1.
\end{align*}
Solving the system of equations yields:
\begin{align}
\qquad P_{0,0}=1-\frac{p+q+\epsilon}{2}, \qquad P_{0,1}=\frac{q+\epsilon-p}{2},\qquad P_{1,0}= \frac{p+\epsilon-q}{2},\qquad P_{1,1}=\frac{p+q-\epsilon}{2}.
\label{eq:lem31}
\end{align}
With the following constraint on the variables:
\begin{align*}
 p+\epsilon\geq q,\qquad p+q\geq\epsilon,\qquad q+\epsilon\geq p, \qquad p+q+\epsilon\leq 2.
\end{align*}
We have:
\begin{align}
&\frac{\mathbb{E}_{X,Y}(gh)}{\mathbb{E}_X^{\frac{1}{2}}(g^2)\mathbb{E}_Y^\frac{1}{2}(h^2)}
=\frac{\alpha\gamma\left(P_{0,0}-P_{0,1}\frac{(1-q)}{q}-P_{1,0}\frac{(1-p)}{p}+P_{1,1}\frac{(1-q)(1-p)}{pq}\right)}
{\alpha\gamma\left(\left((1-p)+\frac{(1-p)^2}{p}\right)^\frac{1}{2}\left((1-q)+\frac{(1-q)^2}{q}\right)^\frac{1}{2}\right)}\nonumber\\
& =\frac{P_{0,0}-P_{0,1}\frac{(1-q)}{q}-P_{1,0}\frac{(1-p)}{p}+P_{1,1}\frac{(1-q)(1-p)}{pq}}
{(\frac{1-p}{p})^\frac{1}{2}(\frac{1-q}{q})^\frac{1}{2}}\nonumber\\
& =\frac{P_{0,0}pq-P_{0,1}(1-q)p-P_{1,0}(1-p)q+P_{1,1}(1-q)(1-p)}
{\left(pq(1-p)(1-q)\right)^\frac{1}{2}}\nonumber\\ 
&\stackrel{\eqref{eq:lem31}}{=}\frac{(1-\frac{p+q+\epsilon}{2})pq-(\frac{q+\epsilon-p}{2})(1-q)p-(\frac{p+\epsilon-q}{2})(1-p)q+(\frac{p+q-\epsilon}{2})(1-q)(1-p)}
{\left(pq(1-p)(1-q)\right)^\frac{1}{2}}\nonumber\\
&=\frac{pq+(\frac{p+q}{2})\left((1-p)(1-p)-pq\right)+(\frac{q-p}{2})\left(q(1-p)-p(1-q))\right)}{\left(pq(1-p)(1-q)\right)^\frac{1}{2}}+\nonumber\\
&\frac{\frac{\epsilon}{2}(pq+p(1-q)+q(1-p)+(1-p)(1-q))}
{\left(pq(1-p)(1-q)\right)^\frac{1}{2}}\nonumber\\
&=\frac{pq+\frac{p+q}{2}(1-p-q)-\frac{p-q}{2}(q-p)-\frac{\epsilon}{2}}
{\left(pq(1-p)(1-q)\right)^\frac{1}{2}}\nonumber\\
&=\frac{p+q-2pq-\epsilon}
{2\left(pq(1-p)(1-q)\right)^\frac{1}{2}}.\label{eq:lem32}
\end{align}
We calculate the optimum point by taking partial derivatives:

\begin{align}
 &\frac{\delta}{\delta p}\frac{\mathbb{E}_{X,Y}(gh)}{\mathbb{E}_X^{\frac{1}{2}}(g^2)\mathbb{E}_Y^\frac{1}{2}(h^2)}=0 \Rightarrow 
 2(1-2q){\left(pq(1-p)(1-q)\right)^\frac{1}{2}}-\frac{(1-2p)}{\sqrt{p(1-p)}}\sqrt{q(1-q)}(p+q-2pq-\epsilon)=0\nonumber\\
& \stackrel{(a)}{\Rightarrow} 2(1-2q)p(1-p)-(1-2p)(p+q-2pq-\epsilon)=0\nonumber\\
 &\Rightarrow 2p(1-p)(1-2q)-p(1-2p)(1-2q)-(1-2p)q+(1-2p)\epsilon=0\nonumber\\
 &\Rightarrow p(1-2q)-(1-2p)q+(1-2p)\epsilon=0\nonumber\\
&\Rightarrow p-q+(1-2p)\epsilon=0.
\label{eq:sys1}
\end{align}
Where in (a) we have used $p,q\notin\{0,1\}$ to multiply by $\sqrt{pq(1-p)(1-q)}$. Taking the partial derivative with respect to $q$, by similar calculations we get:
\begin{align}
 &\frac{\delta}{\delta q}\frac{\mathbb{E}_{X,Y}(gh)}{\mathbb{E}_X^{\frac{1}{2}}(g^2)\mathbb{E}_Y^\frac{1}{2}(h^2)}=0 \to q-p+(1-2q)\epsilon\label{eq:sys2}.
\end{align}
In order for \eqref{eq:sys1} and \eqref{eq:sys2} to be satisfied simultaneously, we must have $\epsilon=0, p=q$, or $\epsilon=p+q=1$, or $p=q=\frac{1}{2}$. For $\epsilon\notin \{0,1\}$, we must have $p=q=\frac{1}{2}$ in which case the value in \eqref{eq:lem32} is:
 
\begin{align*}
 \frac{\mathbb{E}_{X,Y}(gh)}{\mathbb{E}_X^{\frac{1}{2}}(g^2)\mathbb{E}_Y^\frac{1}{2}(h^2)}=1-2\epsilon.
\end{align*}
This completes the proof of the Lemma.
\end{proof}
   Using equations \eqref{Eq:init1} and \eqref{eq:intemid1} we get: 
\begin{align*}
 \mathbb{E}_X(\tilde{e}\tilde{f})\leq  \sum_{\mathbf{i}}C_{\mathbf{i}}\mathbf{P}^{\frac{1}{2}}_{\mathbf{i}}\mathbf{Q}^{\frac{1}{2}}_{\mathbf{i}}.
\end{align*}
\textbf{Step 2:} We use the results from step one to derive a bound on $P(e\neq f)$. Define $a\triangleq P(e(X^n)=1,f(Y^n)=1)$, $b\triangleq P(e(X^n)=0,f(Y^n)=1)$, $c\triangleq P(e(X^n)=1,f(Y^n)=0)$, and $d\triangleq P(e(X^n)=0, f(Y^n)=0)$, then
 
\begin{align}
\label{eq:var}
& \mathbb{E}_{X^n,Y^N}(\tilde{e}(X^n)\tilde{f}(Y^n))=a(1-s)(1-r)-bs(1-r)-c(1-s)r+dsr,
\end{align}
 
We write this equation in terms of $\sigma\triangleq P(f\neq g)$, s, and r using the following relations:
\begin{align*}
&1)\quad a+c=s, \qquad 2) \quad b+d=1-s,\qquad  3)\quad  a+b=r,\qquad  4)\quad  c+d=1-r, \qquad 5)\quad  b+c=\sigma.
\end{align*}
 Solving the above we get:
\begin{align}
 a=\frac{s+r-\sigma}{2}, \qquad b=\frac{r+\sigma-s}{2},\qquad c=\frac{s-r+\sigma}{2},\qquad d=1-\frac{s+r+\sigma}{2}.
\label{eq:intermed}
\end{align}
 
We replace $a,b,c$, and $d$ in \eqref{eq:var} by their values in $\eqref{eq:intermed}$:
\begin{align*}
&\frac{\sigma}{2}\geq (\frac{s+r}{2})(1-s)(1-r)+(\frac{s-r}{2})s(1-r)+(\frac{r-s}{2})(1-s)r+sr(1-\frac{s+r}{2})-\sum_{\mathbf{i}}C_\mathbf{i}\mathbf{P}_{\mathbf{i}}^{\frac{1}{2}}\mathbf{Q}_{\mathbf{i}}^{\frac{1}{2}}\\
&\Rightarrow \sigma\geq s+r-2rs-2\sum_{\mathbf{i}}C_\mathbf{i}\mathbf{P}_{\mathbf{i}}^{\frac{1}{2}}\mathbf{Q}_{\mathbf{i}}^{\frac{1}{2}}
\Rightarrow \sigma \geq (\sqrt{s(1-r)}-\sqrt{r(1-s)})^2+2\sqrt{s(1-s)r(1-r)}-2\sum_{\mathbf{i}}C_\mathbf{i}\mathbf{P}_{\mathbf{i}}^{\frac{1}{2}}\mathbf{Q}_{\mathbf{i}}^{\frac{1}{2}}\\
&\Rightarrow \sigma \geq 2\sqrt{s(1-s)r(1-r)}-2\sum_{\mathbf{i}}C_\mathbf{i}\mathbf{P}_{\mathbf{i}}^{\frac{1}{2}}\mathbf{Q}_{\mathbf{i}}^{\frac{1}{2}}
\end{align*}
 On the other hand $\mathbb{E}_X(\tilde{e}^2)=s(1-s)=\sum_{\mathbf{i}}\mathbf{P}_{\mathbf{i}}$, where the last equality follows from the fact that  $\tilde{e}_{\mathbf{i}}$'s are uncorrelated. This proves the lower bound. Next we use the lower bound to derive the upper bound.
\\\textbf{Step 3:} The upper-bound can be derived by considering the function $h(Y^n)$ to be the complement of $f(Y^n)$ (i.e. $h(Y^n)\triangleq 1\oplus_2 f(Y^n)$.) In this case $P({h}(Y^n)=1)=P(f(Y^n)=0)=1-r$. The corresponding real function for $h(Y^n)$ is:
\begin{align*}
 \tilde{h}(Y^n)=
\begin{cases}
 r  \qquad  &\quad \text{if } h(Y^n)=1,\\
-(1-r) \qquad &\quad \text{if }  h(Y^n)=0,
\end{cases}
=
\begin{cases}
 r  \qquad  &\quad \text{if } f(Y^n)=0,\\
-(1-r) \qquad&\quad \text{if }   f(Y^n)=1,
\end{cases}
\Rightarrow \tilde{h}(Y^n)=-\tilde{f}(Y^n).
\end{align*} 
So, $\tilde{h}(Y^n)=-\sum_{\mathbf{i}}\tilde{f}_{\mathbf{i}}$. Using the same method as in the previous step, we have:
\begin{align*}
 \mathbb{E}_{X^n,Y^n}(\tilde{e}\tilde{h})=-\mathbb{E}_{X^n,Y^n}(\tilde{e}\tilde{f})\leq \sum_{\mathbf{i}}C_{\mathbf{i}}\mathbf{P}^{\frac{1}{2}}_{\mathbf{i}}\mathbf{Q}^{\frac{1}{2}}_{\mathbf{i}} \Rightarrow  P(e(X^n)\neq h(Y^n)) \geq2\sqrt{\sum_{\mathbf{i}}\mathbf{P}_{\mathbf{i}}}\sqrt{\sum_{\mathbf{i}}\mathbf{Q}_{\mathbf{i}}}-2\sum_{\mathbf{i}}C_\mathbf{i}\mathbf{P}_{\mathbf{i}}^{\frac{1}{2}}\mathbf{Q}_{\mathbf{i}}^{\frac{1}{2}} 
\end{align*}
On the other hand $P(e(X^n)\neq h(Y^n))=P(e(X^n)\neq 1\oplus f(Y^n))=P(e(X^n)= f(Y^n))=1-P(e(X^n)\neq f(Y^n))$. So, 
\begin{align*}
&1-P(e(X^n)\neq f(Y^n)) \geq 2\sqrt{\sum_{\mathbf{i}}\mathbf{P}_{\mathbf{i}}}\sqrt{\sum_{\mathbf{i}}\mathbf{Q}_{\mathbf{i}}}-2\sum_{\mathbf{i}}C_\mathbf{i}\mathbf{P}_{\mathbf{i}}^{\frac{1}{2}}\mathbf{Q}_{\mathbf{i}}^{\frac{1}{2}} 
\\&\Rightarrow 
P(e(X^n)\neq f(Y^n)) \leq 1- 2\sqrt{\sum_{\mathbf{i}}\mathbf{P}_{\mathbf{i}}}\sqrt{\sum_{\mathbf{i}}\mathbf{Q}_{\mathbf{i}}}+2\sum_{\mathbf{i}}C_\mathbf{i}\mathbf{P}_{\mathbf{i}}^{\frac{1}{2}}\mathbf{Q}_{\mathbf{i}}^{\frac{1}{2}} 
.
\end{align*}
This completes the proof.
\end{proof}

\subsection{Proof of Theorem \ref{th:sec3_non_bin}}

\begin{proof} The proof of Theorem \ref{th:sec3_non_bin} follows similar steps as the proof of Theorem \ref{th:sec3}. The only difference is in the proof of step 1. 
\\\textbf{Step 1:} Let  $ q\triangleq P_X(e(X^n)=1)$, $r\triangleq P_Y(f(Y^n)=1)$. We have:
 \begin{align}
 &\mathbb{E}_{X^n,Y^n}(\tilde{e}\tilde{f})
 \stackrel{(a)}=\mathbb{E}_{X^n,Y^n}\left(\left(\sum_{\mathbf{i}\in\{0,1\}^n}\tilde{e}_{\mathbf{i}}\right)\left(\sum_{\mathbf{k}\in\{0,1\}^n}\tilde{f}_{\mathbf{k}}\right)\right)
 \stackrel{(b)}{=}\sum_{\mathbf{i}\in\{0,1\}^n}\sum_{\mathbf{k}\in\{0,1\}^n}\mathbb{E}_{X^n,Y^n}( \tilde{e}_{\mathbf{i}}\tilde{f}_{\mathbf{k}}).
 \label{Eq:init2}
\end{align}
In (a) we have used Definition \ref{Rem:Dec}, and in (b) we use linearity of expectation.
 Using the fact that $\tilde{e}_{\mathbf{i}}\in \mathcal{G}_{i_1}\otimes \mathcal{G}_{i_2}\otimes\dotsb \otimes \mathcal{G}_{i_n}$ and Lemma \ref{lem:dec_non_bin}, we have:
 
\begin{align}
\tilde{e}_{\mathbf{i}}(X^n)=\sum_{\forall t\in \tau: l_{t}\in [1,|\mathcal{X}|-1]} c_{\mathbf{i},(l_{t})_{t\in\tau}}\prod_{t\in\tau}\tilde{h}_{l_t}(X_{{t}}),
\qquad
 \tilde{f}_{\mathbf{i}}(Y^n)=\sum_{\forall t\in \tau: l_{t}\in [1,|\mathcal{Y}|-1]} d_{\mathbf{i},(l_{t})_{t\in\tau}}\prod_{t\in\tau}\tilde{g}_{l_t}(Y_{{t}}),
\end{align}
where $c_\mathbf{i,(l_{t})_{t\in\tau}}\in \mathbb{R}$, and $\tilde{h}_l(X), l\in \{1,2,\cdots, |\mathcal{X}|-1\}$, and $\tilde{g}_l(Y), l\in \{1,2,\cdots, |\mathcal{Y}|-1\}$ are a basis for $\mathcal{I}_{X,1}$ and $\mathcal{I}_{Y,1}$, respectively.
We have:

\begin{align}
&\nonumber \mathbb{E}_{X^n,Y^n}( \tilde{e}_{\mathbf{i}}\tilde{f}_{\mathbf{k}})
\stackrel{(a)}= \mathbb{E}_{X^n,Y^n}( \tilde{e}_{\mathbf{i}}\tilde{f}_{\mathbf{i}})\mathbbm{1}(\mathbf{i}=\mathbf{k})\\
&\nonumber\stackrel{(b)}{=}  \mathbbm{1}(\mathbf{i}=\mathbf{k})  \mathbb{E}_{X^n}( \tilde{e}_{\mathbf{i}}\mathbb{E}_{Y^n|X^n}(\tilde{f}_{\mathbf{i}}|X^n))
\stackrel{(c)}{\leq}  \mathbbm{1}(\mathbf{i}=\mathbf{k})\mathbb{E}^\frac{1}{2}_{X^n}( \tilde{e}^2_{\mathbf{i}})\mathbb{E}^\frac{1}{2}_{X^n}( \mathbb{E}^2_{Y^n|X^n}(\tilde{f}_{\mathbf{i}}|X^n))
\\& =\mathbbm{1}(\mathbf{i}=\mathbf{k})\mathbf{P}^{\frac{1}{2}}_{\mathbf{i}}\mathbb{E}^\frac{1}{2}_{X^n}( \mathbb{E}^2_{Y^n|X^n}(\tilde{f}_{\mathbf{i}}|X^n)),
\label{eq:dec_inter_0}
\end{align}
(a) follows by the same arguments as the ones in step 1 of the proof of Theorem \ref{th:sec3}, (b) follows from the law of total expectation and the fact that $e_{\mathbf{i}}$ is a function of $X^n$. In (c) we have used the Cauchy-Schwarz inequality. It only remains to find bounds on $\mathbb{E}_{X^n}( \mathbb{E}^2_{Y^n|X^n}(\tilde{f}_{\mathbf{i}}|X^n))$ which are functions of $\mathbf{Q}_{\mathbf{i}}$, $\psi$, and $N_{\mathbf{i}}$. Let $(i_1,i_2,\cdots,i_{N_{\mathbf{i}}})$ be the indices for which the elements of $\mathbf{i}$ are equal to one. Note that:
\begin{align}
\mathbb{E}_{Y^n|X^n}(\tilde{f}_{\mathbf{i}}|X^n)
&\nonumber=\mathbb{E}_{Y_{\mathbf{i}}|X_{\mathbf{i}}}(\tilde{f}_{\mathbf{i}}|X_{\mathbf{i}})
 =\mathbb{E}_{Y_{i_{N_{\mathbf{i}}}}|X_{i_{N_{\mathbf{i}}}}}\left(\mathbb{E}_{Y_{\mathbf{i}-\mathbf{i}_{N_{\mathbf{i}}}}|X_{\mathbf{i}-\mathbf{i}_{N_{\mathbf{i}}}}} \left( \tilde{f}_{\mathbf{i}}|X_{\mathbf{i}-\mathbf{i}_{N_{\mathbf{i}}}}\right)|X_{i_{N_{\mathbf{i}}}}\right)
\\&=\mathbb{E}_{Y_{i_{N_{\mathbf{i}}}}|X_{i_{N_{\mathbf{i}}}}}\left(\mathbb{E}_{Y_{i_{N_{\mathbf{i}}-1}}|X_{i_{N_{\mathbf{i}}-1}}}\left(\cdots \left( \mathbb{E}_{Y_{i_1}|X_{i_1}}\left(\tilde{f}_{\mathbf{i}}|X_{i_1}\right)|X_{i_2}\right)\cdots\right)|X_{i_{N_{\mathbf{i}}}}\right),
\label{eq:smooth}
\end{align}
 where the first equality follows from the fact that $f_\mathbf{i}$ is a function of $Y_{\mathbf{i}}$. The rest of the equalities follow from the discrete and memoryless properties of the input. For ease of notation define the following projection operators for $1\leq i\leq n$:
 \begin{align*}
 \Pi_{X_i}:&\mathcal{I}_{Y,i}\to \mathcal{I}_{X,i},\\
& h(Y_i)\mapsto \mathbb{E}_{Y_i|X_i}(h(Y_i)).
\end{align*}
$\Pi_{X_i}$ can be interpreted as the projector of zero-mean functions of the random variable $Y_i$ onto zero-mean functions of the random variable $X_i$. We can rewrite Equation \eqref{eq:smooth} as follows:
\begin{equation}
\mathbb{E}_{Y^n|X^n}(\tilde{f}_{\mathbf{i}}|X^n)=
\Pi_{X_{i_{N_{\mathbf{i}}}}}\circ \Pi_{X_{i_{N_{\mathbf{i}}-1}}}\circ\cdots \circ \Pi_{X_{i_1}} (f_{\mathbf{i}}).
\end{equation}
We find bounds on $\mathbb{E}_{X^n}( \mathbb{E}^2_{Y^n|X^n}(\tilde{f}_{\mathbf{i}}|X^n))$ as follows:
\begin{align}
&\nonumber\mathbb{E}_{X^n}( \mathbb{E}^2_{Y^n|X^n}(\tilde{f}_{\mathbf{i}}|X^n))
\\\nonumber&= \mathbb{E}_{X^n}\left(\left(\Pi_{X_{i_{N_{\mathbf{i}}}}}\circ \Pi_{X_{i_{N_{\mathbf{i}}-1}}}\circ\cdots \circ \Pi_{X_{i_1}} (f_{\mathbf{i}})\right)^2\right)
\\\nonumber&\stackrel{(a)}{\leq}\mathbf{Q}_\mathbf{i}||\Pi_{X_{i_{N_{\mathbf{i}}}}}\circ \Pi_{X_{i_{N_{\mathbf{i}}-1}}}\circ\cdots \circ \Pi_{X_{i_1}} ||
\\\nonumber&\stackrel{(b)}{=}\mathbf{Q}_\mathbf{i}||\Pi_{X_{i_{N_{\mathbf{i}}}}}||\cdot||\Pi_{X_{i_{N_{\mathbf{i}}-1}}}||\cdots ||\Pi_{X_{i_1}}||
\\&\stackrel{(c)}{=}\mathbf{Q}_\mathbf{i}||\Pi_{X_1}||^n,
\label{eq:dec_inter_1}
\end{align}
where in (a) the operation norm is defined as $||\Pi||=\sup_e \mathbb{E}(\Pi^2(e))$ where the supremum is taken over all zero-mean functions $e$ with unit variance. (b) follows from the discrete memoryless property of the inputs. Finally, (c) holds since the source elements are identically distributed. On the other hand, we have:
\begin{align}
 \nonumber\psi&=sup_{h,g\in \mathcal{L}} \mathbb{E}_{X_1,Y_1}(h(X_1)g(Y_1))\\
\nonumber &=sup_{h,g\in \mathcal{L}} \mathbb{E}_{X_1} (h(X_1)\mathbb{E}_{Y_1|X_1}(g(Y_1)|X_1))
 \\\nonumber& \stackrel{(a)}{=} sup_{g\in \mathcal{L}} \mathbb{E}^\frac{1}{2}_{X_1} (h^2(X_1)  \mathbb{E}^\frac{1}{2}_{X_1}(\mathbb{E}^2_{Y_1|X_1}(g(Y_1)|X_1))
 \\\nonumber&\stackrel{(b)}{=} sup_{g\in \mathcal{L}} \mathbb{E}^\frac{1}{2}_{X_1}(\mathbb{E}^2_{Y_1|X_1}(g(Y_1)|X_1))
\\&\stackrel{(c)}= ||\Pi_{X_1}||,
\label{eq:dec_inter_2}
\end{align}
where $\mathcal{L}$ is the set of all pairs of functions $g(X)$ and $h(Y)$ with zero mean which have unit variance. (a) follows from the Cauchy-Schwarz inequality and the fact that equality is satisfied by taking $g(X_1)=c\mathbb{E}_{Y_1|X_1}(h(Y_1)|X_1)$ where the constant $c$ is chosen properly, so that $g(X_1)$ has unit variance. The quality (b) holds since $h(X_1)$ has unit variance, and (c) holds by the definition of operator norm. Combining equations \eqref{eq:dec_inter_0}, \eqref{eq:dec_inter_1}, \eqref{eq:dec_inter_2} we have:
\begin{align*}
 \mathbb{E}_{X^n,Y^n}( \tilde{e}_{\mathbf{i}}\tilde{f}_{\mathbf{k}})
 \leq \mathbbm{1}(\mathbf{i}=\mathbf{k}) \psi^{N_{\mathbf{i}}}\mathbf{P}_{\mathbf{i}}^{\frac{1}{2}}\mathbf{Q}_{\mathbf{i}}^{\frac{1}{2}}.
\end{align*}
The rest of the proof follows by the exact same arguments as in steps 2 and 3 in the proof of Theorem \ref{th:sec3}.
\end{proof}
\subsection{Proof of Lemma \ref{lem:SLCE_def}}
\begin{proof}
 We provide an outline of the proof that the three properties in Definition \ref{def:SLCE} are satisfied:
\\ 1) Let $B_n(x^n)$ be as follows:
\begin{align*}
 B_n(x^n)=\{\tilde{x}^n|\exists u^n: (x^n,u^n), (\tilde{x}^n,u^n)\in A_{\epsilon}^n(X,U)\}.
\end{align*}
Following the notation in  \cite{csiszarbook}, let $\mathcal{V}(x^n)$ be the set of all conditional types of sequences $\tilde{x}^n$ given $x^n$, and let $T_v(x^n)$ be the set of all sequences $\tilde{x}^n$ which have the conditional type $v\in \mathcal{V}(x^n)$ with respect to the sequence $x^n$.  
 Then:
\begin{align*}
 |B_n(x^n)|= \sum_{v\in \mathcal{V}(x^n)}|B(x^n)\bigcap T_{v}(x^n)|.
\end{align*}
Note that $|B(x^n)\bigcap T_{v}(x^n)|\neq 0$ if and only if there exists a joint conditional type $\tilde v_{U,\tilde{X}|x^n}$ such that $|\tilde{P}_{U,X}-P_{U,X}|<\epsilon$ and $|\tilde{P}_{U,\tilde{X}}- P_{UX}|<\epsilon$ and $\tilde{P}_{U|X}=v$, where $\tilde{P}_{U,X,\tilde{X}}$ is the joint type of the sequences in $\tilde v_{U,\tilde{X}|x^n}$ with $x^n$. As a result we have: 
\begin{align*}
 |B_n(x^n)|= \sum_{\substack{v\in \mathcal{V}(x^n)
 \\ \exists \tilde v_{U,\tilde{X}|x^n}: 
 \\ |\tilde{P}_{U,X}-P_{U,X}|<\epsilon, |\tilde{P}_{U,\tilde{X}}- P_{UX}|<\epsilon
 }}|B(x^n)\bigcap T_{v}(x^n)|.
\end{align*}

By standard type analysis arguments we conclude that:
\begin{align*}
 |B_n(x^n)|\leq  2^{\max n(H(\tilde{X}|X)+\delta_n)},
\end{align*}
where the maximum is taken over all distributions $P_{U,X,\tilde{X}}$ such that $P_{U,X}=P_{U,\tilde{X}}$, and $\delta_n$ is a sequence of positive numbers which converges to 0 as $n\to\infty$. Since all of the sequences in $B_{n}(X^n)$ are typical we have:
\begin{align*}
P(\tilde{X}^n\in B_{n}(x^n))\approx \frac{|B_n(x^n)|}{|A_{\epsilon}^n(x^n)|}\approx 2^{-n(I(\tilde{X};X)-\delta)}\triangleq 2^{-n\delta_X}.
\end{align*}
Next we show that for $\tilde{x}^n\notin B_n(x^n)$:
 \begin{equation} \label{eq:2neq}
 (1-2^{-n\delta_X})P(\underline{E}(x^n))P(\underline{E}(\tilde{x}^n))<P(\underline{E}(x^n),\underline{E}(\tilde{x}^n))<(1+2^{-n\delta_X})P(\underline{E}(x^n))P(\underline{E}(\tilde{x}^n)) .
\end{equation}
Note that by our construction $x^n$ is mapped to a sequence in $\mathcal{C}\bigcap A_{\epsilon}^n(U|x^n)$ randomly and uniformly. So:
\begin{align*}
& P(\underline{E}(x^n)=c^n|\mathcal{C})=\frac{\mathbbm{1}\left(c^n\in \mathcal{C}\bigcap A_{\epsilon}^n(U|x^n)\right)}{|\mathcal{C}\bigcap A_{\epsilon}^n(U|x^n)|}
 = \frac{\mathbbm{1}\left(c^n\in \mathcal{C}\bigcap A_{\epsilon}^n(U|x^n)\right)}{|\mathcal{C}\bigcap A_{\epsilon}^n(U|x^n)- \{c^n\}|+1}
 \\& \Rightarrow  P(\underline{E}(x^n)=c^n)=P\left(c^n\in \mathcal{C}\cap A_{\epsilon}^n(U|x^n)\right)\mathbb{E}_{\mathcal{C}}\left(\frac{1}{|\mathcal{C}\bigcap A_{\epsilon}^n(U|x^n)- \{c^n\}|+1}\right).
\end{align*}
By a similar argument for $\tilde{x}^n\notin B_n(x^n)$ we have:

\begin{align*} 
&P(\underline{E}(x^n)=c^n, \underline{E}(\tilde{x}^n)=\tilde{c}^n)=\\&P\left(c^n\in \mathcal{C}\cap A_{\epsilon}^n(U|x^n)
,\tilde{c}^n\in \mathcal{C}\cap A_{\epsilon}^n(U|\tilde{x}^n)\right)
\mathbb{E}_{\mathcal{C}}\left(\frac{1}{|\mathcal{C}\bigcap A_{\epsilon}^n(U|x^n)- \{c^n\}|+1}
\frac{1}{|\mathcal{C}\bigcap A_{\epsilon^n(U|\tilde{x}^n)}- \{\tilde{c}^n\}|+1}
\right).
\end{align*}
We show the following bound:
\begin{align}\label{eq:neq}
 P\left(c^n\in \mathcal{C}\cap A_{\epsilon}^n(U|x^n)
,\tilde{c}^n\in \mathcal{C}\cap A_{\epsilon}^n(U|\tilde{x}^n)\right)
\leq P\left(c^n\in \mathcal{C}\cap A_{\epsilon}^n(U|x^n)\right)P\left(
\tilde{c}^n\in \mathcal{C}\cap A_{\epsilon}^n(U|\tilde{x}^n)\right)(1+2^{-n\delta_X}). 
\end{align}
Assume that $c^n\in A_{\epsilon}^n(U|{x}^n)$, and $\tilde{c}^n\in A_{\epsilon}^n(U|\tilde{x}^n)$, otherwise the two sides are equal to 0. The following equalities hold by the construction algorithm:
\begin{align*}
& P\left(c^n\in \mathcal{C}\cap A_{\epsilon}^n(U|x^n)\right)= P\left(c^n\in \mathcal{C}\right)=\frac{|\mathcal{C}|}{|A_{\epsilon}^n(U)|},
\\& P\left(\tilde{c}^n\in \mathcal{C}\cap A_{\epsilon}^n(U|\tilde{x}^n)\right)= P\left(\tilde{c}^n\in \mathcal{C}\right)=\frac{|\mathcal{C}|}{|A_{\epsilon}^n(U)|},
\\& P\left(c^n\in \mathcal{C}\cap A_{\epsilon}^n(U|x^n)
,\tilde{c}^n\in \mathcal{C}\cap A_{\epsilon}^n(U|\tilde{x}^n)\right)=
P\left(c^n\in \mathcal{C}
,\tilde{c}^n\in \mathcal{C}\right)
=\frac{{|\mathcal{C}|\choose 2}}{{A_{\epsilon}^n(U) \choose 2}}.
\end{align*}
Using the above it is straightforward to check that the bound in (\ref{eq:neq}) holds. Similarly, it follows that 
\begin{align*}
&\mathbb{E}_{\mathcal{C}}\left(\frac{1}{|\mathcal{C}\bigcap A_{\epsilon}^n(U|x^n)- \{c^n\}|+1}
\frac{1}{|\mathcal{C}\bigcap A_{\epsilon^n(U|\tilde{x}^n)}- \{\tilde{c}^n\}|+1}
\right)\leq
\\& \mathbb{E}_{\mathcal{C}}\left(\frac{1}{|\mathcal{C}\bigcap A_{\epsilon}^n(U|x^n)- \{c^n\}|+1}\right)\mathbb{E}_{\mathcal{C}}\left(
\frac{1}{|\mathcal{C}\bigcap A_{\epsilon^n(U|\tilde{x}^n)}- \{\tilde{c}^n\}|+1}
\right)(1+2^{-n\delta_X}).
\end{align*}
Multiplying the two bound recovers the right-hand side of the inequality in (\ref{eq:2neq}). The left-hand side can be shown by similar arguments. 

2) As $n$ becomes large, the $i$th output element $E_i(X^n)$ is correlated with the input sequence $X^n$ only through the $i$th input element $X_i$:
\begin{align*}
 &\forall \delta>0, \exists n\in \mathbb{N}: m>n
 \Rightarrow \forall x^m\in \{0,1\}^m , v \in \{0,1\},
 \\& |P_{\mathscr{S}}(E_i(X^m)=v|X^m=x^m)-P_{\mathscr{S}}(E_i(X^m)=v|X_i=x_i)|<\delta. 
\end{align*}

The proof is as follows: For a fixed quantization function $\underline{e}:\{0,1\}^m\to\{0,1\}^m$, $\underline{e}(X^m)$ is a function of $X^m$. However,
 without the knowledge that which encoding function is used, $E_i(X^m)$ is related to $X^m$ only through $X_i$. In other words, averaged over all encoding functions, the effects of the rest of the elements diminishes. We provide a proof of this statement below:
 
First, we are required to provide some definitions relating to the joint type of pairs of sequences. For binary strings $u^m, x^m$, define $N(a,b|u^m,x^m)\triangleq
|\{j|u_j=a, x_j=b\}|$, that is the number of indices $j$ for which the value of the pair $(u_j,x_j)$ is $(a,b)$. For $s,t\in \{0,1\}$, define $l_{s,t}\triangleq N(s,t|u^m,x^m)$, the vector $(l_{0,0}, l_{0,1}, l_{1,0}, l_{1,1})$ is called the joint type of $(u^m,x^m)$. For fixed $x^m$ The set of sequences $T_{l_{0,0}, l_{0,1}, l_{1,0}, l_{1,1}}=\{u^m| N(s,t|u^m,x^m)=l_{s,t}, s,t\in\{0,1\}\}$, is the set of vectors which have joint type $(l_{0,0}, l_{0,1}, l_{1,0}, l_{1,1})$ with the sequence $x^m$. Fix $m, \epsilon>0$, and define $\mathcal{L}_{\epsilon,n}\triangleq \{(l_{0,0}, l_{0,1}, l_{1,0}, l_{1,1}): |\frac{l_{s,t}}{m}-P_{U,X}(s,t)|<\epsilon, \forall s,t\}$. Then for the conditional typical set $A_{\epsilon}^n(U|x^m)$ defined above we can write
\begin{align*}
 A_{\epsilon}^n(U|x^m)=\bigcup_{(l_{0,0}, l_{0,1}, l_{1,0}, l_{1,1})\in \mathcal{L}_{\epsilon,n}} T_{l_{0,0}, l_{0,1}, l_{1,0}, l_{1,1}}.
\end{align*}
The type of $x^m$, denoted by $(l_0,l_1)$ is defined in a similar manner.
 Since $E_i(X^m)$ are chosen uniformly from the set $A_{\epsilon}^n(U|x^m)$, we have:
\begin{align*}
&P_{\mathscr{S}}(E_i(X^m)=v|X^m=x^m)=\frac{|\{u^m| u_1=v, u^m\in  A_{\epsilon}^m(U|x^m)\}|}{|\{u^m| u^m\in  A_{\epsilon}^m(U|x^m)\}|} 
\\&=\frac{\sum_{(l_{0,0}, l_{0,1}, l_{1,0}, l_{1,1})\in  \mathcal{L}_{\epsilon,n}} |\{u^m| u_1=v, u^m\in  T_{l_{0,0}, l_{0,1}, l_{1,0}, l_{1,1}}\}|}{\sum_{(l_{0,0}, l_{0,1}, l_{1,0}, l_{1,1})\in  \mathcal{L}_{\epsilon,n}} |\{u^m|  u^m\in  T_{l_{0,0}, l_{0,1}, l_{1,0}, l_{1,1}}\}|}\\
&=\frac{\sum_{(l_{0,0}, l_{0,1}, l_{1,0}, l_{1,1})\in  \mathcal{L}_{\epsilon,n}} {{l_{x_1}-1} \choose {l_{u_1,x_1}-1} }{{l_{\bar{x}_1}} \choose {l_{u_1,\bar{x}_1}} }}{\sum_{(l_{0,0}, l_{0,1}, l_{1,0}, l_{1,1})\in  \mathcal{L}_{\epsilon,n}}{{l_{x_1}} \choose {l_{u_1,x_1}} }{{l_{\bar{x}_1}} \choose {l_{u_1,\bar{x}_1}} }}
\\&=\frac{\sum_{(l_{0,0}, l_{0,1}, l_{1,0}, l_{1,1})\in  \mathcal{L}_{\epsilon,n}} \frac{(l_{x_1}-1)!}{
(l_{u_1,x_1}-1)!(l_{x_1}-l_{u_1,x_1})!}
 \frac{l_{\bar{x}_1}!}{
l_{u_1,\bar{x}_1}!(l_{\bar{x}_1}-l_{u_1,\bar{x}_1})!}}
{\sum_{(l_{0,0}, l_{0,1}, l_{1,0}, l_{1,1})\in  \mathcal{L}_{\epsilon,n}} \frac{l_{x_1}!}{
l_{u_1,x_1}!(l_{x_1}-l_{u_1,x_1})!}
 \frac{l_{\bar{x}_1}!}{
l_{u_1,\bar{x}_1}!(l_{\bar{x}_1}-l_{u_1,\bar{x}_1})!}}
\\&\stackrel{(a)}{=}\frac{\sum_{(l_{0,0}, l_{0,1}, l_{1,0}, l_{1,1})\in  \mathcal{L}_{\epsilon,n}} l_{u_1,x_1} \frac{1}{
l_{u_1,x_1}!(l_{x_1}-l_{u_1,x_1})!}
 \frac{1}{
l_{u_1,\bar{x}_1}!(l_{\bar{x}_1}-l_{u_1,\bar{x}_1})!}}
{l_{x_1}\sum_{(l_{0,0}, l_{0,1}, l_{1,0}, l_{1,1})\in  \mathcal{L}_{\epsilon,n}} \frac{1}{
l_{u_1,x_1}!(l_{x_1}-l_{u_1,x_1})!}
 \frac{1}{
l_{u_1,\bar{x}_1}!(l_{\bar{x}_1}-l_{u_1,\bar{x}_1})!}}
\\&\stackrel{(b)}{\Rightarrow} \frac{P_{U,X}(u_1,x_1)-\epsilon}{P_X(x_1)+\epsilon} \leq P_{\mathscr{S}}(E_i(X^m)=v|X^m=x^m) \leq  \frac{P_{U,X}(u_1,x_1)+\epsilon}{P_X(x_1)-\epsilon}
\\& \Rightarrow \exists m,\epsilon>0: |P_{\mathscr{S}}(E_i(X^m)=v|X^m=x^m)-P_{U|X}(u_1|x_1)|\leq \delta.
\end{align*}
In (a), we use the fact that for fixed $x^m$, $(l_{x_1},l_{\bar{x}_1})$ is fixed to simplify the numerators. In (b) we have used that for jointly typical $\epsilon$-sequences $(u^m,x^m)$, $l_{u_1,x_1}\in [n(P_{U,X}(u_1,x_1)-\epsilon), n(P_{U,X}(u_1,x_1)+\epsilon)]$, and $l_{x_1}\in  [n(P_{X}(x_1)-\epsilon), n(P_{X}(x_1)+\epsilon)]$.

3) The encoder is insensitive to permutations. Due to typicality encoding the probability that a vector $x^n$ is mapped to $y^n$ depends only on their joint type and is equal to the probability that $\pi(x^n)$ is mapped to $\pi(y^n)$.
\end{proof}
\subsection{Proof of Proposition \ref{prop:5.1}}
\begin{proof}

Fix $k,k'\in \mathbb{N}$. Define the permutation ${\pi_{k\to k'}}\in S_n$ as the permutation which switches the $k$th and $k'$th elements and fixes all other elements. Also, let $\mathscr{E}$ be the set of all mappings $e:\{0,1\}^n\to \{0,1\}^n$. 
\begin{align*}
  P_{\mathscr{S}}\left(\sum_{{\mathbf{i}}:N_{\mathbf{i}}\leq m, \mathbf{i}\neq \mathbf{i}_k} \mathbf{P}_{k,{\mathbf{i}}}>\gamma\right)
 &=\sum_{\underline{e}\in \mathscr{E}}
 P_{\mathscr{S}}\left(\underline{e}\right)\mathbbm{1}\left(\sum_{{\mathbf{i}}:N_{\mathbf{i}}\leq m, \mathbf{i}\neq  \mathbf{i}_k|\underline{e}} \mathbf{P}_{k,{\mathbf{i}}}>\gamma\right)
\\&\stackrel{(a)}{=}\sum_{\underline{e}\in \mathscr{E}}
 P_{\mathscr{S}}\left(\underline{e}_{\pi_{k\to k'}}\right)\mathbbm{1}\left(\sum_{{\mathbf{i}}:N_{\mathbf{i}}\leq m, \mathbf{i}\neq \mathbf{i}_k} \mathbf{P}_{k,{\mathbf{i}}}>\gamma|\underline{e}\right)\\
  \\&\stackrel{(b)}{=}\sum_{\underline{g}\in \mathscr{E}}
 P_{\mathscr{S}}\left(\underline{g}\right)\mathbbm{1}\left(\sum_{{\mathbf{i}}:N_{\mathbf{i}}\leq m, \mathbf{i}\neq \mathbf{i}_k} \mathbf{P}_{{\pi_{k\to k'}}{k},{\pi_{k\to k'}}{\mathbf{i}}}>\gamma|\underline{g}\right)
\\&  =\sum_{\underline{g}\in \mathscr{E}}
 P_{\mathscr{S}}\left(\underline{g}\right)\mathbbm{1}\left(\sum_{{\underline{l}}:N_{\underline{l}}\leq m, \underline{l}\neq{\pi_{k\to k'}}\mathbf{i}_k} \mathbf{P}_{k',{\underline{l}}}>\gamma|\underline{g}\right)
\\&  =P_{\mathscr{S}}\left(\sum_{{\mathbf{i}}:N_{\mathbf{i}}\leq m, \mathbf{i}\neq \mathbf{i}_{k'}} \mathbf{P}_{k',{\mathbf{i}}}>\gamma\right),
  \end{align*}
where in (a) we have used property 3) in Definition \ref{def:SLCE}, and in (b) we have defined $\underline{g}\triangleq \underline{e}_{{\pi_{k\to k'}}}$ and used ${\pi^2_{k\to k'}}=1$.
  
\end{proof}

\subsection{Proof of Claim \ref{claim:expo}}
\begin{proof}
\begin{align*}
& \mathbb{E}_{\tilde{E}, X_{\mathbf{i}}}(\mathbb{E}^2_{X^n|X_{\mathbf{i}}}(\tilde{E}|X_{\mathbf{i}}))=\sum_{x_{\mathbf{i}}, \tilde{e}}P(x_{\mathbf{i}})P( \tilde{e})(\sum_{x_{\sim \mathbf{i}}}P(x_{\sim \mathbf{i}})\tilde{e}(x^n))^2\\
 &=\sum_{x_{\mathbf{i}}, \tilde{e}}P(x_{\mathbf{i}})P( \tilde{e})\sum_{x_{\sim \mathbf{i}}}\sum_{y^n: y_{\mathbf{i}}=x_{\mathbf{i}}}P(x_{\sim \mathbf{i}})P(y_{\sim \mathbf{i}})\tilde{e}(x^n)\tilde{e}(y^n)\\
 &=\sum_{x^n}P(x^n)\sum_{y^n:y_{\mathbf{i}}=x_{\mathbf{i}}}P(y_{\sim \mathbf{i}})\mathbb{E}_{\tilde{E}}(\tilde{E}(x^n)\tilde{E}(y^n))\\
 &=\sum_{x^n}P(x^n)\sum_{y^n:y_{\mathbf{i}}=x_{\mathbf{i}}, y^n \in B_n(x^n)}P(y_{\sim \mathbf{i}})\mathbb{E}_{\tilde{E}}(\tilde{E}(x^n)\tilde{E}(y^n))+ 
 \sum_{x^n}P(x^n)\sum_{y^n:y_{\mathbf{i}}=x_{\mathbf{i}}, y^n\notin B_n(x^n)}P(y_{\sim \mathbf{i}})\mathbb{E}_{\tilde{E}}(\tilde{E}(x^n)\tilde{E}(y^n))\\
 &\stackrel{(a)}{\leq} \sum_{x^n}P(x^n)\sum_{y^n:y_{\mathbf{i}}=x_{\mathbf{i}}, y^n \in B_n(x^n)}P(y_{\sim \mathbf{i}})+
 \sum_{x^n}P(x^n)\sum_{y^n:y_{\mathbf{i}}=x_{\mathbf{i}}, y^n\notin B_n(x^n)}P(y_{\sim \mathbf{i}})\mathbb{E}_{\tilde{E}}(\tilde{E}(x^n)\tilde{E}(y^n))\\
 &=P(Y^n\in B_{n}(X^n)|Y_{\mathbf{i}}=X_{\mathbf{i}})+ \sum_{x^n}P(x^n)\sum_{y^n:y_{\mathbf{i}}=x_{\mathbf{i}}, y^n\notin B_n(x^n)}P(y_{\sim \mathbf{i}})\mathbb{E}_{\tilde{E}}(\tilde{E}(x^n)\tilde{E}(y^n))\\
 &\stackrel{(b)}{=}O(e^{-n\delta_X})+\sum_{x^n}P(x^n)\sum_{y^n:y_{\mathbf{i}}=x_{\mathbf{i}}, y^n\notin B_n(x^n)}P(y_{\sim \mathbf{i}})\mathbb{E}_{\tilde{E}}(\tilde{E}(x^n))\mathbb{E}_{\tilde{E}}(\tilde{E}(y^n))\\
 &\leq O(e^{-n\delta_X})+P(Y^n\in B_{n}(X^n)|Y_{\mathbf{i}}=X_{\mathbf{i}})+
\sum_{x_{\mathbf{i}}}P(x_{\mathbf{i}})\sum_{x_{\sim \mathbf{i}}}\sum_{y^n:y_{\mathbf{i}}=x_{\mathbf{i}}}P(x_{\sim \mathbf{i}})P(y_{\sim \mathbf{i}})\mathbb{E}_{\tilde{E}}(\tilde{E}(x^n))\mathbb{E}_{\tilde{E}}(\tilde{E}(y^n))\\
&=O(e^{-n\delta_X})+\mathbb{E}_{X_{\mathbf{i}}}(\mathbb{E}^2_{\tilde{E}, {X^n|X_{\mathbf{i}}}}(\tilde{E}|X_{\mathbf{i}})).\\
\end{align*}
\end{proof}
In (a) we use the fact that $\tilde{E}\leq 1$ by definition, in (b) follows from property 1) in Definition \ref{def:SLCE}.
%
%

\bibliographystyle{unsrt}
  \bibliography{FLReferences}   

%
 \end{document}